\documentclass[12pt]{article}
\usepackage{epsfig}
\usepackage{graphicx}
\usepackage{makeidx}
\usepackage{multicol}
\usepackage{subeqn}
\usepackage{times}
\usepackage{amssymb}  

\usepackage{crop}
\makeindex

\newtheorem{definition}{Definition}
\newtheorem{lemma}{Lemma}
\newtheorem{theorem}{Theorem}
\newtheorem{proposition}{Proposition}
\newtheorem{corollary}{Corollary}

\def\endproof{$\blacksquare$\bigskip}

\def\coeff{{\rm coeff}}

\def\<{\leqslant}           
\def\>{\geqslant}           

\def\d{\partial}
\def\wh{\widehat}
\def\wt{\widetilde}

\def\Re{{\rm Re}}   
\def\Im{{\rm Im}}   

\def\cH{{\cal H}}   
\def\mA{{\mathbb A}}    
\def\mR{{\mathbb R}}    

\def\Tr{{\rm Tr}}       
\def\rT{{\rm T}}        
\def\diam{\diamond}       

\def\bE{{\bf E}}    

\def\[[[{[\![\![}
\def\]]]{]\!]\!]}

\def\bra{{\langle}}
\def\ket{{\rangle}}

\def\re{{\rm e}}        
\def\rd{{\rm d}}        

\def\fa{{\mathfrak a}}

\def\bJ{{\bf J}}

\def\br{{\bf r}}
\def\x{\times}
\def\od{\odot}
\def\ox{\otimes}

\def\bone{{\bf 1}}

\def\cI{{\mathcal I}}
\def\cP{{\mathcal P}}

\def\cT{{\mathcal T}}

\def\mS{{\mathbb S}}

\def\eps{\epsilon}
\def\Ups{\Upsilon}

\def\rprod{\mathop{\prod\!\!\!\!\!\!\!\!\!\!\longrightarrow}}

\def\diag{\mathop{\rm diag}}    

\begin{document}
\title{\bf Gaussian Stochastic  Linearization for Open Quantum Systems
Using Quadratic Approximation of Hamiltonians\footnote{This work is supported by the Australian Research Council.}}
\author{
    Igor G. Vladimirov$^{\dagger}$, \qquad Ian R. Petersen%
    \thanks{
        School of Engineering and Information Technology,
        University of New South Wales at the Australian Defence Force Academy,
        Canberra, ACT 2600,
        E-mail: {\small\tt igor.g.vladimirov@gmail.com, i.r.petersen@gmail.com}.
        }
}

\date{}

\maketitle

%

\paragraph{Abstract}
{\small This paper extends the energy-based version of the stochastic linearization method, known for classical nonlinear systems, to open quantum  systems with canonically commuting dynamic variables governed by quantum stochastic differential equations with non-quadratic Hamiltonians. The linearization proceeds by approximating the actual Hamiltonian of the quantum system by a quadratic function of its observables which corresponds to the Hamiltonian of a quantum harmonic oscillator. This approximation is carried out in a mean square optimal sense with respect to a Gaussian reference quantum state and leads to a self-consistent linearization procedure where the mean vector and quantum covariance matrix of the system observables evolve in time according to the effective linear dynamics.  We demonstrate the proposed Hamiltonian-based Gaussian linearization
for the quantum Duffing oscillator whose Hamiltonian is a quadro-quartic polynomial of the momentum and position operators. The results of the paper are applicable to the design of suboptimal controllers and filters  for nonlinear quantum systems.}

\thispagestyle{empty}

\section{Introduction}

A wide class of models for open quantum systems \cite{BP_2006,GZ_2004}, that is, quantum-mechanical objects interacting with the environment, is provided by dynamical systems whose state variables
are canonically commuting self-adjoint operators on a  Hilbert space. In  the Heisenberg picture, these system observables evolve in time according to quantum stochastic differential equations (QSDEs) \cite{P_1992}. Such QSDEs, which are dual to quantum master equations for density operators in the Schr\"{o}dinger picture  \cite[Chapter 3]{BP_2006}, are driven by a quantum Wiener process to take into account the coupling between the  environment (regarded as a memoryless heat bath of quantum harmonic oscillators) and the internal  dynamics which the system would have in isolation from the surroundings.
These internal dynamics are completely specified by the system Hamiltonian, which is a self-adjoint operator on the underlying Hilbert space,  usually representable as a function of the system observables.

In particular, quadratic system Hamiltonians correspond to quantum harmonic oscillators whose behaviour lends itself to complete analysis due to linearity of the resulting QSDEs in contrast to the general nonlinear case. Linear open quantum systems are being actively researched to develop  quantum analogues of classical control schemes, including the $\cH_{\infty}$, risk-sensitive and linear quadratic Gaussian control approaches (see, for example, \cite{EB_2005,J_2004,JNP_2008,MP_2010,NJP_2009,VP_2011a,VP_2011b} and references therein). Such models are also employed in quantum optics which is considered to be one of  possible platforms for implementing the quantum computer \cite[Section 7.4]{NC_2000}.

The present paper is aimed at a quantum-mechanical version of the stochastic linearization (SL) technique whose origins date back to \cite{B_1953,C_1963,K_1956} (see also \cite{BW_1998,C_2006,R_2007} and references therein). SL is concerned with a classical SDE whose drift term is a nonlinear function of the state vector.
%
%
The principal idea of SL is to approximate the drift by an affine function of the state variables whose coefficients are computed using a mean square criterion with respect to a probability distribution.
 This reference distribution,  which is intended to mimic the actual probability distribution of the state vector, is usually chosen to be Gaussian, although non-Gaussian approximations (such as, for example,  in \cite{C_2006,R_2007}) are also utilized. The Gaussian reference measure leads to an effective linear SDE which approximates the actual nonlinear dynamics. A salient feature of this SDE is that its coefficients depend nonlinearly (through integral operators with Gaussian kernels) on the mean value  and  covariance matrix of the state vector, which are in turn governed by  linear ordinary differential equations (ODEs) (including the Lyapunov ODE for the covariance matrix) involving those coefficients.\footnote{This resembles the McKean-Vlasov SDE (from the kinetic theory of plasma) whose drift depends on the probability density function of the state vector propagated by the Fokker-Planck-Kolmogorov equation associated with the SDE, thus leading to a nonlinear parabolic partial differential equation  \cite{M_1966}.}  This provides a self-consistent procedure for linearizing  the dynamics.

An alternative energy-based version \cite{ZEZ_1990}  of the SL technique was aimed originally at structural engineering problems of random vibrations with a \emph{potential} nonlinear restoring force.  Rather than directly linearizing the nonlinearity, this approach employs a mean square criterion in order to approximate the force potential by a quadratic function of the displacement vector (corresponding to an ideal spring). It is this variant of the classical SL that is particularly suitable for our purposes.  We adapt it to the quantum-mechanical  setting by solving the problem of minimizing the mean square deviation between the actual non-quadratic system Hamiltonian of the quantum  system and a general quadratic function of its observables. The  solution involves the second and higher-order  mixed moments, which, in  a Gaussian quantum state \cite{P_2010} (see also \cite[pp. 118--122]{GZ_2004}), are completely specified by the mean vector and the quantum covariance matrix of the system observables through Wick's theorem \cite[p.~122]{M_2005}.

For a class of open quantum systems,  whose coupling with the external heat bath variables in the total Hamiltonian is bilinear,  the quadratic approximation of the system Hamiltonian leads to a linear QSDE of an open quantum harmonic oscillator which is amenable to comprehensive analysis. In particular, similarly to the classical linear systems \cite{KS_1972}, the quantum covariance matrix of the observables of this effective oscillator satisfies a Lyapunov ODE. Moreover, such linearization respects the physical realizability (PR) conditions \cite[Theorem 3.4 on p. 1790]{JNP_2008}, which makes it suitable for coherent quantum control \cite{NJP_2009}. We demonstrate the approach for the quantum Duffing oscillator \cite{C_1988,PT_2006} with a quadro-quartic Hamiltonian. The proposed Hamiltonian-based quantum Gaussian linearization technique is applicable to the development of suboptimal controllers and filters for nonlinear quantum systems since it offers a recipe to deal with the ``curse of dimensionality'' of the information state, similar to projective quantum filtering \cite{VM_2005}.

The paper is organised as follows. Section~\ref{sec:quad_ham} specifies the class of quadratic Hamiltonians. Section~\ref{sec:quad_approx} describes the approximation of an arbitrary Hamiltonian by a quadratic Hamiltonian, optimal in the mean square sense. Section~\ref{sec:Gauss} specializes the computations for a Gaussian quantum state. Section~\ref{sec:self} describes the Hamiltonian-based self-consistent  linearization for open quantum systems. Section \ref{sec:Duffing} demonstrates the linearization procedure for the quantum Duffing oscillator. Section~\ref{sec:conclusion} provides concluding remarks. Long proofs and subsidiary material are given in Appendices.

%

\section{Quadratic Hamiltonians in canonically commuting variables}\label{sec:quad_ham}

Suppose $x_1, \ldots, x_n$ are quantum observables (that is, self-adjoint operators on an underlying separable Hilbert space  $\cH$ with an inner product $\bra\varphi \mid \psi \ket$)\footnote{To avoid confusion, we use a different notation $\bra \cdot, \cdot \ket$  for other inner products, for example, the Frobenius inner product of matrices.}  which satisfy canonical commutation relations (CCRs)
\begin{equation}
\label{theta}
    [x_j, x_k] = i
    \theta_{jk}\cI,
    \qquad
    1\< j,k\< n.
\end{equation}
Here, $i:= \sqrt{-1}$ is the imaginary unit,
$[A, B]:= AB-BA$ is the commutator of operators,   and $\Theta:= (\theta_{jk})_{1\<j, k\< n}$ is a real antisymmetric CCR matrix of order $n$ (the space of such matrices is denoted by $\mA_n$). Also, $\cI$ denotes the identity operator which carries out the ampliation  of entries of the matrix $\Theta$ to the space of linear operators on $\cH$ and will be omitted for brevity, so that  (\ref{theta}) can be written in a vector-matrix form as
\begin{equation}
\label{Theta}
    [x,x^{\rT}]
    :=
    \big([x_j,x_k]\big)_{1\< j,k\< n}
    =
    i
    \Theta,
\end{equation}
where the observables are assembled into a vector $x:= (x_j)_{1\< j\< n}$. Unless indicated otherwise, vectors are organised as columns. The transpose $(\cdot)^{\rT}$ applies to vectors and matrices with operator-valued entries as if the latter were scalars. In particular, the CCRs hold for self-adjoint operators which are representable as linear combinations of annihilation and creation operators $\fa_1, \ldots, \fa_{\nu}$ and $\fa_1^{\dagger}, \ldots, \fa_{\nu}^{\dagger}$, where $\nu := n/2$ and $n$ is assumed to be even. Such are, for example, the quantum-mechanical position and momentum operators $q$ and $p:= -i
\d_q$ with $[q,p] = i
$ and CCR matrix
\begin{equation}
\label{bJ}
    \bJ
    :=
    {\small\left[\begin{array}{rc}
        0& 1\\
        -1 & 0
    \end{array}\right]},
\end{equation}
which spans the space $\mA_2$. 
The associated annihilation and creation operators \cite[pp. 90--91]{S_1994}
\begin{equation}
\label{faa}
    \fa
    :=
    (q+ip)/\sqrt{2},
    \qquad
    \fa^{\dagger}
    :=
    (q-ip)/\sqrt{2}
\end{equation}
satisfy $[\fa, \fa^{\dagger}] = 1$.
Now, with a scalar $a \in \mR$, a vector $b:= (b_j)_{1\< j \< n} \in \mR^n$ and a  real symmetric matrix $R:= (r_{jk})_{1\< j,k\< n}$ of order $n$ (the space of such matrices is denoted by $\mS_n$), we associate a self-adjoint operator
\begin{equation}
\label{HabR}
    H_{a,b,R}
     :=
    a
    +
    b^{\rT}x
    +
    x^{\rT}Rx\big/2
     =
    a
    +
    \sum_{j=1}^n
    \Big(
        b_k
        +
        \frac{1}{2}
        \sum_{j=1}^n
        r_{jk}x_j
    \Big)
    x_k
\end{equation}
on the Hilbert space $\cH$. The operator $H_{a,b,R}$, which is parameterized linearly  by the triple $(a, b, R) \in \mR\x \mR^n \x\mS_n$, is the Hamiltonian of a quantum harmonic oscillator with state variables $x_1, \ldots, x_n$. Although the constant term $a$ in (\ref{HabR}) has no influence on the system dynamics, it is retained to preserve the generality of $H_{a,b,R}$ as a quadratic polynomial of the system observables with real coefficients. If the system is isolated from the environment, its Heisenberg dynamics are described by the ODEs
\begin{eqnarray}
\nonumber
    \dot{x}_{\ell}
    & = &
    i[H_{a,b,R}, x_{\ell}]
     =
    i
    \Big(
        \sum_{j=1}^n
        b_j
        [x_j,x_{\ell}]
        +
        \frac{1}{2}
        \sum_{j,k=1}^n
        r_{jk}
        [x_jx_k, x_{\ell}]
    \Big)\\
\label{xelldot}
&=&
    -\sum_{j=1}^n
    b_j \theta_{j\ell}
    -\frac{1}{2}
    \sum_{j,k=1}^n
    r_{jk}
    (\theta_{k\ell}x_j+\theta_{j\ell}x_k)
=
    \sum_{j=1}^n
    \theta_{\ell j}
    \Big(
        b_j
        +
        \sum_{k=1}^{n}
        r_{jk}
        x_k
    \Big),
\end{eqnarray}
where the commutator identity $    [AB, C]
    =
    A [B, C] + [A, C]B
$
from \cite[Eq. (3.50) on p. 38]{M_1998}
is applied to the triple $x_j$, $x_k$, $x_{\ell}$ and use is made of the CCRs from (\ref{theta}) along with the antisymmetry of $\Theta$. In a vector-matrix form, (\ref{xelldot}) can be written as
\begin{equation}
\label{xdot}
    \dot{x}
    =
    i[H_{a,b,R}, x]
    =
    \Theta (b + Rx).
\end{equation}
If $R \succ 0$, the spectrum of the matrix $\Theta R$ is purely imaginary, and the  system is neutrally stable. In general, $R$ is not necessarily positive definite. For example, quantum amplifiers \cite{GZ_2004}, used as active elements in quantum optics, are modelled as inverted oscillators with $R\prec 0$. If $R$ is nonsingular, the effect of $b$ reduces to a constant shift $R^{-1}b$ in $x$, so that (\ref{xdot}) can be written in terms of $y:= x + R^{-1}b$ as $\dot{y} = \Theta R y$. If the system has a non-quadratic Hamiltonian $H$, then, in contrast to  (\ref{xdot}), the right-hand side $i[H,x] := (i[H,x_j])_{1\< j\< n}$ of the Heisenberg dynamics is not affine in the system observables. In this case, $i[H,x]$ can, in principle,  be approximated by $\Theta(b+ R x)$ so as to minimize a mean square deviation between the vectors:
\begin{equation}
\label{ELS}
    \bE(
        \Delta_{b,R}^{\rT} F \Delta_{b,R}
    )
    \longrightarrow
    \min,
    \qquad
    b \in \mR^n,\
    R\in \mS_n.
\end{equation}
Here, $\bE A := \Tr(\rho A)$ denotes the quantum expectation of a linear operator $A$ on the underlying Hilbert space $\cH$ with respect to a density operator $\rho$ (a positive semi-definite self-adjoint operator  on $\cH$ with unit trace $\Tr \rho =1$) which specifies the quantum state \cite[p. 51]{P_1992}. 
Also,
\begin{equation}
\label{Delta}
    \Delta_{b,R}
    :=
    i
    [H,x] - \Theta(b + R x)
\end{equation}
is a vector of ``residuals'' which depends affinely on $b$ and $R$, and $F$ is a complex positive definite Hermitian matrix of order $n$. Such linearization of $i[H,x]$ can be regarded as a quantum version of the weighted least squares from the classical linear regression analysis \cite{RTSH_2010}. The weight matrix $F$, which governs the minimization problem (\ref{ELS})--(\ref{Delta}), influences the optimal values of $b$ and $R$, and its particular choice requires additional consideration. 
We will therefore take a different approach to linearizing the system dynamics, through a quadratic approximation of the Hamiltonian itself, as a quantum counterpart to the energy-based variant of SL given in \cite{ZEZ_1990}.

\section{Mean square optimal quadratic approximation of Hamiltonians}\label{sec:quad_approx}

Consider the mean square optimal approximation of the system Hamiltonian $H$ by a quadratic Hamiltonian $H_{a,b,R}$ from (\ref{HabR}):
\begin{equation}
\label{Q}
    Q(\alpha,\beta,R)
     :=
    \bE((H -H_{a,b,R})^2)    \\
     =
    \bE((\eta -h_{\alpha,\beta,R})^2)
    \longrightarrow \min.
\end{equation}
Here, we have introduced a different parameterization of the quadratic Hamiltonian
\begin{equation}
\label{h}
    H_{a,b,R}
    =
    \alpha + \beta^{\rT} \xi + \xi^{\rT} R\xi/2
        =:
    h_{\alpha, \beta, R}
\end{equation}
through ``centering'' the observables of the system in the reference quantum state:
\begin{equation}
\label{xieta}
    \xi
    :=
    (\xi_j)_{1\< j \< n}
    :=
    x -\bE x,
    \qquad
    \eta
    :=
    H - \bE H.
\end{equation}
The new parameters     $\alpha
      :=
    a + b^{\rT} \bE x + (\bE x)^{\rT} R \bE x/2 - \bE H$
 and $b + R\bE x$ are bijectively related to the old ones $a$ and $b$ from (\ref{HabR}), with the matrix $R$ remaining the same.
To compute the optimal values of $\alpha$, $\beta$, $R$ which minimize $Q(\alpha, \beta, R)$, we will use the real parts of the following  mixed central moments of the actual Hamiltonian $H$ and the system observables $x_1, \ldots, x_n$:
\begin{equation}
\label{Hx_Hxx_xx}
    \eps_j
    :=
    \Re
    \bE(\eta\xi_j),
    \qquad
    \gamma_{jk}
    :=
    \Re\bE(\eta\xi_j\xi_k),
    \qquad
    \sigma_{jk}
     :=
    \Re \bE(\xi_j \xi_k),
\end{equation}
\begin{equation}
\label{xxx_xxxx}
    \tau_{jk\ell}
     :=
    \Re \bE(\xi_j\xi_k\xi_{\ell}),
    \qquad
    \varphi_{jk\ell m}
     :=
    \Re \bE(\xi_j\xi_k\xi_{\ell}\xi_m).
\end{equation}
Several remarks  are in order on the matrices $\Gamma := (\gamma_{jk})_{1\< j,k\< n}$ and $\Sigma := (\sigma_{jk})_{1\< j,k\< n}$ and the tensors $\cT:= (\tau_{jk\ell})_{1\< j,k,\ell\< n}$ and $\Phi:= (\varphi_{jk\ell m})_{1\< j,k,\ell,m\< n}$ defined by (\ref{Hx_Hxx_xx})--(\ref{xxx_xxxx}). Since the system observables satisfy CCRs,  then in view of Lemma~\ref{lem:three} from Appendix~\ref{sec:three}, the matrix $\Gamma  = \Re \bE (\eta \xi\xi^{\rT})$ is symmetric and the tensor $\cT$ is totally symmetric.
%
In what follows, $\cT$ is identified with a linear operator acting from $\mS_n$ to $\mR^n$ as
\begin{equation}
\label{cTmap}
    \cT(R)
    :=
    \Big(
        \sum_{k,\ell=1}^{n}
        \tau_{jk\ell}r_{k\ell}
    \Big)_{1\< j\< n},
    \qquad
    R:=(r_{k\ell})_{1\< k,\ell\< n}\in \mS_n.
\end{equation}
The matrix $\Sigma = \Re \bE(\xi\xi^{\rT})$ is symmetric and positive semi-definite as the real part of the quantum  covariance matrix of observables. However, the tensor $\Phi$ from (\ref{xxx_xxxx}) is only guaranteed to be symmetric with respect to  reversing the order of its subscripts. 
Indeed, similarly to (\ref{33}) from Appendix~\ref{sec:three},
$$
    \overline{\bE(\xi_j\xi_k\xi_{\ell}\xi_m)}
    =
    \overline{\Tr(\rho\xi_j\xi_k\xi_{\ell}\xi_m)}
    =
    \Tr (\xi_m\xi_{\ell}\xi_k\xi_j \rho)
     =
    \bE(\xi_m\xi_{\ell}\xi_k\xi_j),
$$
with $\overline{(\cdot)}$ the complex conjugate,
and hence, the quantum expectations on the opposite sides have equal real parts, that is, $\varphi_{jk\ell m} = \varphi_{m\ell k j}$. The fact that $\varphi_{jk\ell m}$ is, in general, not invariant even under transpositions of its neighbouring  subscripts follows from  the identities
\begin{eqnarray}
\nonumber
    \varphi_{jk\ell m}
    -
    \varphi_{kj\ell m}
     & = &
    \Re \bE([\xi_j,\xi_k]\xi_{\ell} \xi_m)
    =
    \Re(i\theta_{jk} \bE(\xi_{\ell} \xi_m))\\
\label{phitrans1}
     & = &
    \Re(i\theta_{jk} (\sigma_{\ell m} + i\theta_{\ell m}/2))
     =
    -\theta_{jk}\theta_{\ell m}/2
    =
    \varphi_{jk\ell m}
    -
    \varphi_{jkm\ell}
\end{eqnarray}
and a similar relationship
\begin{equation}
\label{phitrans2}
    \varphi_{jk\ell m}
    -
    \varphi_{j\ell km}
    =
    -\theta_{jm}\theta_{k\ell}/2,
\end{equation}
which are established by using the CCRs (\ref{theta}) and the definitions (\ref{Hx_Hxx_xx}), (\ref{xxx_xxxx}).
Now, consider a partial symmetrization $\Psi:=(\psi_{jk\ell m})_{1\< j,k,\ell,m\< n}$ of $\Phi$ whose entries are defined by
\begin{eqnarray}
\nonumber
    \psi_{jk\ell m}
     & := &
    (\varphi_{jk\ell m} + \varphi_{jkm\ell}+\varphi_{kj\ell m} + \varphi_{kjm\ell}\\
\label{Psi}
      & &+
    \varphi_{\ell m jk} + \varphi_{m\ell jk}+\varphi_{\ell m kj} + \varphi_{m\ell kj})/8,
\end{eqnarray}
which involves only eight of the 24 possible permutations of the subscripts $j$, $k$, $\ell$, $m$. We will identify $\Psi$ with a self-adjoint operator on the Hilbert space $\mS_n$ (with the Frobenius inner product of matrices $\bra X,Y\ket:= \Tr(XY)$ inherited from $\mR^{n\x n}$) defined by
\begin{equation}
\label{Psimap}
    \Psi(R)
    :=
    \Big(
    \sum_{\ell,m=1}^{n}
    \psi_{jk\ell m}
    r_{\ell m}
    \Big)_{1\< j,k\< n},
    \qquad
    R:=(r_{\ell m})_{1\< \ell, m\< n}\in \mS_n.
\end{equation}
The significance of $\Psi$ as the partial symmetrization of $\Phi$ (with the latter being regarded as a linear operator on $\mR^{n\x n}$, defined similarly) is that
\begin{equation}
\label{PhiPsi}
    \bE((\xi^{\rT} R \xi)^2)
    =
    \bra
        R, \Phi(R)
    \ket
    =
    \bra
        R,\Psi(R)
    \ket,
    \qquad
    R \in \mS_n,
\end{equation}
which, in fact, can be used as an equivalent definition of $\Psi$. Since the quantum expectation of a squared observable is always nonnegative, (\ref{PhiPsi}) implies that the operator $\Psi$ is positive semi-definite ($\Psi\succcurlyeq 0$).
By using the identity $\bE(AB) =\overline{\bE(BA) }$ for observables $A$, $B$,
it follows that the mean square criterion  (\ref{Q}) is a convex quadratic function:
\begin{equation}
\label{Q1}
    Q(\alpha,\beta,R)
    =
    \bE(\eta^2)
    -
    2\Re \bE(\eta h_{\alpha,\beta,R})
    +
    \bE (h_{\alpha,\beta,R}^2).
\end{equation}
In view of (\ref{h})--(\ref{Hx_Hxx_xx}), the second term on the right-hand side of (\ref{Q1}) does not depend on $\alpha$ and is linear with respect to  $\beta$ and $R$:
\begin{equation}
\label{Q2}
    \Re \bE(\eta h_{\alpha,\beta,R})
    =
    \eps^{\rT}\beta
    +
    \bra
        \Gamma,
        R
    \ket/2.
\end{equation}
By a similar reasoning, (\ref{h})--(\ref{xxx_xxxx}) imply that the rightmost term in (\ref{Q1}) is a positive semi-definite quadratic form
\begin{eqnarray}
\nonumber
    \bE (h_{\alpha,\beta,R}^2)
    & = &
    \alpha^2
    +
    \alpha
    \bra \Sigma, R\ket
    +
    \beta^{\rT}\Sigma \beta\\
\label{Q3}
    & & +\beta^{\rT}\cT(R)
    +
    \frac{\bra R, \Psi(R)\ket}{4}
    =
    \bra
        \zeta, \Pi(\zeta)
    \ket,
    \qquad
    \zeta
    :=
    {\small\left[
    \begin{array}{c}
        \alpha\\
        \beta\\
        R
    \end{array}
    \right]},
\end{eqnarray}
which is specified by a self-adjoint operator $\Pi$
on the Hilbert space $\mR \x \mR^n\x \mS_n$ (with the inherited inner product $\bra \zeta, \zeta'\ket := \alpha\alpha' + \beta^{\rT}\beta' + \bra R, R'\ket$) as
\begin{equation}
\label{Pi}
    \Pi(\zeta)
     :=
    {\small\left[\begin{array}{c}
        \alpha  + \bra \Sigma, R \ket/2\\
        \Sigma \beta + \cT(R)/2\\
        \Sigma \alpha/2 + \cT^{\dagger}(\beta)/2 + \Psi(R)/4
    \end{array}\right]}
    =
    {\small\left[\begin{array}{ccc}
        1 & 0 & \bra \Sigma, \cdot\ket/2\\
        0 & \Sigma & \cT/2\\
        \Sigma /2 & \cT^{\dagger}/2 & \Psi/4
    \end{array}\right]}
    \zeta.
\end{equation}
Here, use is made of the operators $\cT$ and $\Psi$ from (\ref{cTmap}), (\ref{Psimap}), and the adjoint operator $\cT^{\dagger}: \mR^n\to \mS_n$ maps a vector $\beta:= (\beta_j)_{1\< j\< n}$ to a matrix $\cT^{\dagger}(\beta):= \big(\sum_{k,\ell=1}^n\tau_{jk\ell}\beta_j\big)_{1\< k, \ell\< n}$.
The  symbolic matrix representation of $\Pi$ in (\ref{Pi}) can be identified with the real part of the matrix of second moments of the triple $(1, \xi, \xi\xi^{\rT}/2)$.    By a generalized version of the Schur complement condition  of positive definiteness \cite[Theorems  7.7.6, 7.7.7 on pp. 472--474]{HJ_2007}, the invertibility of the positive semi-definite operator $\Pi$ is equivalent to $\Sigma \succ 0$ and $G\succ 0$, where $G$ is a self-adjoint operator on $\mS_n$ defined by
\begin{equation}
\label{G}
    G
    :=
    \Psi
    -
    {\small\left[\begin{array}{cc}
        \Sigma &
        \cT^{\dagger}
    \end{array}
    \right]}
    {\small\left[\begin{array}{cc}
        1 & 0\\
        0 & \Sigma
    \end{array}
    \right]}^{-1}
    {\small\left[\begin{array}{c}
        \bra
            \Sigma,
            \cdot
        \ket \\
        \cT
    \end{array}
    \right]}
    =
    \Psi - \Sigma\bra \Sigma, \cdot\ket - \cT^{\dagger}\Sigma^{-1} \cT.
\end{equation}
Up to a factor of 4, the operator $G$ is  the Schur complement of the block ${\scriptsize\left[\begin{array}{cc}1 & 0\\ 0 & \Sigma\end{array}\right]}$ in (\ref{Pi}).


\begin{theorem}
\label{th:bestquad}
Suppose the operator $\Pi$, defined by (\ref{Pi}), is invertible. Then
the optimal values of the parameters $\alpha$, $\beta$, $R$, which minimize the function $Q(\alpha, \beta, R)$ in (\ref{Q1}), are computed in terms of (\ref{Hx_Hxx_xx}), (\ref{xxx_xxxx}) and (\ref{G}) as
\begin{eqnarray}
\label{alphadiam}
    \alpha_{\diam}
    &=&
    -\bra\Sigma, R_{\diam}\ket/2, \\
\label{betadiam}
    \beta_{\diam}
    &=&
    \Sigma^{-1}
    (\eps-\cT(R_{\diam})/2),\\
\label{Rdiam}
    R_{\diam}
    &=&
    2G^{-1}
    (\Gamma-\cT^{\dagger}(\Sigma^{-1}\eps)).
\end{eqnarray}
\end{theorem}
\begin{proof}
The optimal values of $\alpha$, $\beta$, $R$ are obtained by equating the Frechet derivatives of the function $Q(\alpha, \beta, R)$ in (\ref{Q1}) to zero. In view of (\ref{Q2}) and (\ref{Q3}), this leads to the system of linear equations
\begin{eqnarray}
\label{dalpha}
    \d_{\alpha} Q
    & = &
    2\alpha + \bra\Sigma, R\ket = 0, \\
\label{dbeta}
    \d_{\beta} Q
    & = &
    2\Sigma \beta + \cT(R)-2\eps
    =0,\\
\label{dR}
    \d_R Q
    & = &
    \Psi(R)/2 + \cT^{\dagger}(\beta) + \alpha \Sigma - \Gamma=0,
\end{eqnarray}
which correspond to the normal equations of the least squares method in the linear regression analysis \cite{RTSH_2010}. If the operator $\Pi\succcurlyeq 0$ in (\ref{Pi}) is invertible (and hence, $\Pi \succ 0$), then the quadratic function $Q$ in (\ref{Q1}) is strictly convex and the system of equations (\ref{dalpha})--(\ref{dR}) has a unique solution. Now, (\ref{alphadiam}) and (\ref{betadiam}) follow  directly from (\ref{dalpha}) and (\ref{dbeta}), while (\ref{Rdiam}) is established by their substitution into (\ref{dR}) and using the invertibility of the operator $G$ from (\ref{G}) which is secured by the condition $\Pi\succ 0$.
\end{proof}

Since the vector $\eps$, the matrices $\Gamma$ and $\Sigma$ and the tensors $\cT$ and $\Psi$, which are associated with the mixed central moments of the Hamiltonian $H$ and the system observables $x_1, \ldots, x_n$, depend on the quantum state, then so also do the parameters $\alpha_{\diam}$, $\beta_{\diam}$,   $R_{\diam}$ of the optimal quadratic approximation of the Hamiltonian.

%
%
%

\section{Approximating the Hamiltonian in a Gaussian quantum state}\label{sec:Gauss}

The system is said to be in a Gaussian quantum state \cite{P_2010}, if the quantum covariance function of the centered vector $\xi$ of system observables from (\ref{xieta}) is given by
\begin{equation}
\label{charfun}
    \bE \re^{i u^{\rT}\xi}
    =
    \re^{- u^{\rT} S u/2}
    =
    \re^{- u^{\rT} \Sigma u/2},\
    \qquad
    u 
    \in \mR^n.
\end{equation}
Here,
\begin{equation}
\label{S}
    S
    :=
    (s_{jk})_{1\< j, k \< n}
    :=
    \bE(\xi\xi^{\rT})
    =
    \Sigma + i \Theta/2,
\end{equation}
is the quantum covariance matrix, which is a complex positive semi-definite Hermitian matrix, with $\Theta$ and $\Sigma$ defined by (\ref{Theta}), (\ref{Hx_Hxx_xx}), and use is made of the antisymmetry of $\Theta$.
By applying Wick's theorem \cite[p.~122]{M_2005}, which is a quantum counterpart to Isserlis' theorem \cite{I_1918} on the mixed central moments of evenly many jointly Gaussian classical random variables in terms of their covariances (see also \cite[Theorem 1.28 on pp.~11--12]{J_1997}), it follows that, in the Gaussian quantum state,
\begin{equation}
\label{Wick}
    \bE(\xi_{j_1} \x \ldots \x \xi_{j_{2r}})
    =
    \sum
    \prod_{\ell = 1}^{r}
    s_{j_{k_{2\ell-1}} j_{k_{2\ell}}}.
\end{equation}
Here, the sum of products of the quantum covariances from (\ref{S}) extends over a class $\cP_r$ of $(2r-1)!!$ permutations $(k_1, \ldots, k_{2r})$ of the integers  $1, \ldots, 2r$ which satisfy $k_{2\ell - 1}< k_{2\ell}$ for every $1\< \ell  \< r$ and $k_1 < k_3 < \ldots < k_{2r-3} < k_{2r-1}$. Such permutations will be referred to as \emph{regular}. There is a one-to-one correspondence between the regular permutations and all
possible partitions $\{\{k_1, k_2\}, \ldots, \{k_{2r-1}, k_{2r}\}\}$
of the set $\{1, \ldots, 2r\}$ into two-element subsets.
Thus,  (\ref{Wick}) allows any mixed moment of even order $2r$ to be computed for the observables $\xi_1, \ldots, \xi_n$  in a Gaussian quantum state in terms of the matrix $S$, whereas all the moments of odd orders in such a state are zero; see Appendix~\ref{sec:Wick}.
In particular, the tensor $\cT$ of third order  moments from (\ref{xxx_xxxx}) vanishes, while application of (\ref{Wick}) to the fourth order mixed moments yields
\begin{equation}
\label{fourth}
    \bE(\xi_j\xi_k\xi_{\ell}\xi_m)
    =
    s_{jk} s_{\ell m}
    +
    s_{j\ell} s_{k m}
    +
    s_{jm} s_{k\ell},
\end{equation}
cf. a similar relation for the annihilation and creation operators (\ref{faa}) in \cite[Eq. (4.4.121) on p. 122]{GZ_2004}. Hence, the real parts of the fourth order moments in (\ref{xxx_xxxx}) take the form
\begin{equation}
\label{xxxxWick}
    \varphi_{jk\ell m}
    =
    \sigma_{jk} \sigma_{\ell m}
    +
    \sigma_{j\ell} \sigma_{k m}
    +
    \sigma_{jm} \sigma_{k\ell}
    -(    \theta_{jk} \theta_{\ell m}
    +
    \theta_{j\ell} \theta_{k m}
    +
    \theta_{jm} \theta_{k\ell})/4,
\end{equation}
which is in accordance with the more general relationships (\ref{phitrans1}), (\ref{phitrans2}).
However, it will be more convenient to compute the partial symmetrization $\Psi$ of $\Phi$ by applying (\ref{fourth}) to the equivalent definition of $\Psi$ in (\ref{PhiPsi}) rather than using the entrywise representation (\ref{xxxxWick}).

\begin{lemma}
\label{lem:mom}
Suppose the system is in a Gaussian quantum state. Then the operator $\Psi$, defined by (\ref{Psi})--(\ref{PhiPsi}) in terms of the tensor $\Phi$ from (\ref{xxx_xxxx}), takes the form
\begin{equation}
\label{Psigauss}
    \Psi(R)
    =
    \Sigma \bra \Sigma, R\ket
    +
    2K(R),
\end{equation}
where $K$ is a positive semi-definite self-adjoint  operator on the space $\mS_n$, defined by
\begin{equation}
\label{K}
    K(R)
    :=
    \Sigma R \Sigma + \Theta R\Theta/4,
\end{equation}
with $\Theta$ and $\Sigma$ defined by (\ref{Theta}) and (\ref{Hx_Hxx_xx}). If the quantum covariance matrix $S$ from (\ref{S}) is nonsingular, then $K\succ 0$.
\end{lemma}

We prove Lemma~\ref{lem:mom} in Appendix~\ref{sec:momproof}. In a particular case, when the system observables commute with each other, that is,  $\Theta = 0$, Lemma~\ref{lem:mom} reduces to the well-known result on the second moment of a quadratic form in jointly Gaussian classical random variables \cite[Lemma~2.3 on p. 204]{M_1978}. If $\Theta \ne 0$, the noncommutative quantum nature of the system observables enters (\ref{K})  through the additional term $\Theta R\Theta/4$ which makes $K$ a special self-adjoint operator of grade two \cite[Section 7]{VP_2011a}.
The above discussion allows Theorem~\ref{th:bestquad} to be concretized for  the Gaussian quantum case as follows.

\begin{theorem}
\label{th:bestquadgauss}
Suppose the mean square deviation $Q(\alpha, \beta, R)$ in  (\ref{Q1}) is associated with a Gaussian quantum state, and the quantum covariance matrix $S$ in (\ref{S}) is nonsingular.
Then
the optimal values of $\alpha$, $\beta$, $R$ in (\ref{alphadiam})--(\ref{Rdiam}) are computed in terms of the mixed central moments (\ref{Hx_Hxx_xx}) and the associated positive definite operator $K$ from (\ref{K}) as
\begin{equation}
\label{alphadiamgauss_betadiamgauss_Rdiamgauss}
    \alpha_{\diam}
    =
    -\bra
        \Sigma,
        R_{\diam}
    \ket/2,
    \qquad
    \beta_{\diam}
    =
    \Sigma^{-1}
    \eps,
    \qquad
    R_{\diam}
    =
    K^{-1}
    (\Gamma).
\end{equation}
\end{theorem}
\begin{proof}
The expressions (\ref{alphadiamgauss_betadiamgauss_Rdiamgauss}) are obtained from (\ref{alphadiam})--(\ref{Rdiam}) by noting that, in the Gaussian quantum state, $\cT=0$ and, in view of Lemma~\ref{lem:mom}, the operator (\ref{G}) takes the form $G=\Psi-\Sigma\bra\Sigma, \cdot\ket = 2K$, where the invertibility of $K$ is ensured  by the assumption that $S\succ 0$.
\end{proof}

Although the operator $K$ in (\ref{K}) is completely specified by the matrices $\Sigma$ and $\Theta$, the optimal parameters  $\beta_{\diam}$ and $R_{\diam}$ in (\ref{alphadiamgauss_betadiamgauss_Rdiamgauss}) will also depend on the mean value $\bE x$ in the Gaussian reference state through the vector $\eps$ and the matrix $\Gamma$. The computation of the inverse operator $K^{-1}$, which is required for (\ref{alphadiamgauss_betadiamgauss_Rdiamgauss}), is described in Appendix~\ref{sec:Kinv} 
where it is also shown that the condition $S\succ 0$ ensures the positiveness of $K^{-1}$ with respect to the convex cone $\mS_n^+$   of real positive semi-definite symmetric matrices of order $n$ in the sense that \begin{equation}
\label{Kpos}
    K^{-1}(\mS_n^+) \subset \mS_n^+.
\end{equation}

\section{Self-consistent quantum Gaussian linearization}\label{sec:self}

Suppose the quantum system interacts with the external heat bath so that the $n$-dimensional vector $X_t$ of its observables at time $t$ is governed by a QSDE
\begin{equation}
\label{dX}
    \rd X_t
    =
    \big(
        i
        [H,X_t]
        -
        BJB^{\rT}\Theta^{-1}X_t/2
    \big)
    \rd t
    +
    B\rd W_t,
\end{equation}
where $H$ is the system Hamiltonian discussed previously, and the CCR matrix $\Theta$ of the system observables is assumed to be nonsingular. This corresponds to a bilinear coupling between the open quantum system and  the bath variables in the total Hamiltonian as quantified by a constant matrix $B \in \mR^{n\x m}$; see, for example, \cite{EB_2005,JNP_2008} for details. Also,
$W_t$ is an $m$-dimensional quantum Wiener process (with $m$ even) which represents the influence of the environment  on the system.  The entries of $W_t$ are self-adjoint operators on a  boson Fock space \cite{P_1992} with the quantum Ito table
\begin{equation}
\label{WW}
    \rd W_t
    \rd W_t^{\rT}
    =
    \Omega\rd t,
    \qquad
    \Omega
    :=
    I_m + iJ/2,
    \qquad
    J =  \bJ\ox I_{m/2},
\end{equation}
where $I_r$ denotes the identity matrix of order $r$,   the matrix $\bJ$ is given by (\ref{bJ}), and $\ox$ is the Kronecker product of matrices, so that $J$ is the CCR matrix of $W_t$  in the sense that
$
    [
        \rd W_t,
        \rd W_t^{\rT}
    ]
    =
    iJ\rd t
$.
Consider the first two moments of the system observables:
\begin{equation}
\label{MS}
    \mu_t := \bE X_t,
    \qquad
    \Sigma_t
    :=
    \Re S_t,
    \qquad
    S_t := \bE(\xi_t\xi_t^{\rT}),
    \qquad
    \xi_t := X_t - \mu_t.
\end{equation}
Although the quantum state of the system is not necessarily Gaussian,  $\mu_t$ and $\Sigma_t$ can be used to compute the parameters of the mean square optimal quadratic approximation $H_{\alpha_{\diam},\beta_{\diam}, R_{\diam}}$ of the actual Hamiltonian $H$ through Theorem~\ref{th:bestquadgauss} as if the state were Gaussian. In this Gaussian reference quantum state, the term $i[H,X_t]$ in (\ref{dX}) can be approximated as
\begin{equation}
\label{Happ}
    i[H,X_t]
    \approx
    i[H_{\alpha_{\diam},\beta_{\diam},R_{\diam}},\, X_t]
=    i[\beta_{\diam}^{\rT} \xi_t + \xi_t^{\rT}R_{\diam}\xi_t/2, \, \xi_t]
    =
    \Theta(\beta_{\diam} + R_{\diam} \xi_t).
\end{equation}
Formal substitution of (\ref{Happ}) into the right-hand side of (\ref{dX})  yields a linear approximation for this QSDE which splits into an approximate ODE for the mean $\mu_t$ and an approximate QSDE for the centered vector $\xi_t$ of system observables from (\ref{MS}):
\begin{equation}
\label{Mdot_dxi}
    \dot{\mu}_t
     \approx
    \Theta \beta_{\diam} - BJB^{\rT} \Theta^{-1} \mu_t/2,
    \qquad
    \rd \xi_t
    \approx
    A_t \xi_t\rd t + B\rd W_t.
\end{equation}
Here, the external field is assumed to be in the vacuum state \cite{P_1992}, and the matrix
\begin{equation}
\label{At}
    A_t
    :=
    \Theta R_{\diam}  - BJB^{\rT}\Theta^{-1}/2
\end{equation}
depends deterministically on time $t$ through $R_{\diam}$ which is completely specified by $\mu_t$, $\Sigma_t$ according to Theorem~\ref{th:bestquadgauss}. Upon averaging, the quantum Ito differential  $\rd (\xi_t \xi_t^{\rT}) = (\rd \xi_t) \xi_t^{\rT} + \xi_t
\rd \xi_t^{\rT} + (\rd \xi_t) \rd \xi_t^{\rT}$, combined  with (\ref{Mdot_dxi}) and (\ref{WW}) according to \cite[Proposition 25.26 on pp. 202--203]{P_1992}, leads to an approximate Lyapunov ODE $\dot{S}_t \approx A_t S_t + S_t A_t^{\rT} + B\Omega B^{\rT}$ for the quantum covariance matrix $S_t$ from (\ref{MS}). If this approximation is regarded as an exact equation
\begin{equation}
\label{Sdot}
    \dot{S}_t
    =
    A_t S_t + S_t A_t^{\rT} + B\Omega B^{\rT},
\end{equation}
its solution satisfies  $S_t= \Sigma_t +i\Theta/2 \succ 0$ for all $t\> 0$, provided $S_0:= \Sigma_0 + i\Theta/2 \succ 0$. Here, we have used the positive semi-definiteness of the matrix $\Omega$ from (\ref{WW}) and the $\succ$-monotinicity of the transition operator of the Lyapunov ODE, with the fact that $A_t$ from (\ref{At}) satisfies
\begin{equation}
\label{PR}
    A_t\Theta + \Theta A_t^{\rT} + BJB^{\rT}=0.
\end{equation}
The latter property, which is equivalent to the  preservation of the CCR matrix $\Theta$ in time, is one of the physical realizability (PR) conditions \cite[Theorem 3.4 on p. 1790]{JNP_2008} describing the dynamic equivalence of the system to an open quantum harmonic oscillator. 
%
 Now, if the first of the equations (\ref{Mdot_dxi}) is also considered as an exact ODE, then its combination  with the real part of (\ref{Sdot}) yields
\begin{equation}
\label{Mdot1_Sigmadot1}
    \dot{\mu}_t
     =
    \Theta \beta_{\diam} - BJB^{\rT} \Theta^{-1} \mu_t/2,
    \qquad
    \dot{\Sigma}_t
     =
    A_t\Sigma_t+\Sigma_t A_t^{\rT} + BB^{\rT}.
\end{equation}
In conjunction with (\ref{At}), the ODEs (\ref{Mdot1_Sigmadot1}) provide a self-consistent set of nonlinear equations for finding $\mu_t$ and $\Sigma_t$ as functions of time, with the nonlinearity coming from the dependence of $\beta_{\diam}$, $R_{\diam}$ in (\ref{alphadiamgauss_betadiamgauss_Rdiamgauss})  on $\mu_t$, $\Sigma_t$. Therefore, although the above ODEs result from an ad hoc approximation, this quantum Gaussian linearization procedure generates a faithful quantum covariance matrix for the system observables and respects the PR conditions in the sense of (\ref{PR}).
A time invariant version of the procedure is obtained by equating the right-hand sides of (\ref{Mdot1_Sigmadot1}) to zero and considering \textit{admissible} solutions $\mu$, $\Sigma$ of the corresponding  algebraic equations
\begin{equation}
\label{Mdot2_Sigmadot2}
    \Theta \beta_{\diam} - BJB^{\rT} \Theta^{-1} \mu/2
    =
    0,
    \qquad
    A\Sigma+\Sigma A^{\rT} + BB^{\rT}
    =
    0,
\end{equation}
for which the matrix
\begin{equation}
\label{Ainfty}
    A
    :=
    \Theta R_{\diam}  - BJB^{\rT}\Theta^{-1}/2,
\end{equation}
obtained from (\ref{At}),
is Hurwitz. If the matrix $B$ is of full row rank, then, in view of $\Omega \succ 0$,  the corresponding solution $S$ of the algebraic Lyapunov  equation $AS+SA^{\rT} + B\Omega B^{\rT} = 0$ as a steady-state version of (\ref{Sdot}) (assuming $A$ Hurwitz)  is nonsingular, thus ensuring the invertibility of the operator $K$, which is essential for the quantum Gaussian linearization through Lemma~\ref{lem:mom} and Theorem~\ref{th:bestquadgauss}. Since the functions  $\beta_{\diam}$ and $R_{\diam}$ can, in general, be of complicated nature, their presence makes  (\ref{Mdot2_Sigmadot2})--(\ref{Ainfty}) a system of nonlinear vector-matrix algebraic equations for which the existence/uniqueness of admissible solutions $\mu$, $\Sigma$ is a nontrivial problem. Nevertheless, the example in the next section demonstrates tractability of this problem for a class of anharmonic oscillators.

\section{Application to the quantum Duffing oscillator}\label{sec:Duffing}

Consider the quantum Duffing oscillator \cite{C_1988,PT_2006}   of unit mass on the real line with the Hamiltonian
\begin{equation}
\label{HD}
    H
    :=
    V(q)
    +
    \frac{p^2}{2},
    \qquad
    V(q)
    :=
    \frac{\omega_0^2q^2}{2} + fq^4,
    \qquad
    x
    :=
    {\small\left[\begin{array}{c}
        q\\
        p
    \end{array}
    \right]}.
\end{equation}
Here, the vector $x$ of system observables
is formed by the position and momentum operators  $q$ and $p$ from Section~\ref{sec:quad_ham}
with the CCR matrix $\Theta := \bJ$ given by (\ref{bJ}). The quadro-quartic polynomial $V(q)$ describes the potential energy, where $\omega_0$ is the harmonic frequency and the coefficient $f$ ``weights''  the quartic part, which is responsible for the anharmonicity of the oscillator.
%
%
Application of the relation $[V(q),p] = V'(q)[q,p] = iV'(q)$, which follows from  the commutator identity of \cite[Eq. (3.51) on p. 39]{M_1998}) and the CCR $[q,p]=i$, leads to
$$
    i[H,x]
    =
    {\small\left[
        \begin{array}{cc}
            p\\
            -\omega_0^2 q - 4 fq^3
        \end{array}
    \right]}
$$
containing a cubic term.
We will now demonstrate the Hamiltonian-based quantum  Gaussian linearization technique of Sections \ref{sec:Gauss} and \ref{sec:self} for the quantum Duffing oscillator.  According to (\ref{betaDuff}) and (\ref{RDuff}) of Appendix~\ref{sec:Duffcalcs}, whose derivation is based on Theorem \ref{th:bestquadgauss} and repeated use of Wick's theorem,
the parameters of the mean square optimal quadratic approximation of the Hamiltonian (\ref{HD}) in a Gaussian reference quantum state are computed as
\begin{equation}
\label{best}
    \beta_{\diam}
    =
    N\mu,
    \
    N:={\small
    \left[\begin{array}{cc}
        \omega_0^2 + 4f(\kappa^2 + 3\sigma_{11}) & 0\\
        0 & 1
    \end{array}\right]
    },
    \
    R_{\diam}
    =
    {\small
    \left[\begin{array}{cc}
        \omega_0^2 + 12f(\kappa^2 + \sigma_{11}) & 0\\
        0 & 1
    \end{array}\right]
    }.
\end{equation}
Here,  $\mu:= \bE x = {\scriptsize\left[\begin{array}{c} \kappa \\ \bE p\end{array}\right]}$, with $\kappa:= \bE q$ and $\sigma_{11} := \bE((q-\kappa)^2)$ the mean and variance of the position operator. Suppose the open quantum Duffing oscillator is governed by the QSDE (\ref{dX}), with $B\in \mR^{2\x m}$ of full row rank.   Since the matrix $\bJ$ spans the space $\mA_2$, there exists a real $\phi$ such that
\begin{equation}
\label{phi}
    BJB^{\rT}=\phi\bJ.
\end{equation}
Hence,  $BJB^{\rT}\Theta^{-1}=\phi I_2$, and the time invariant version of the quantum Gaussian linearization procedure (\ref{Mdot1_Sigmadot1}) reduces to the algebraic equations
\begin{equation}
\label{infMS}
    (\bJ N - \phi I_2/2)\mu  = 0,
    \qquad
    A\Sigma+\Sigma A^{\rT} + C =0,
\end{equation}
where
\begin{equation}
\label{C}
    C:= (c_{jk})_{1\< j,k\< 2}:= BB^{\rT}
\end{equation}
satisfies $C + i\phi\bJ/2 = B\Omega B^{\rT}\succ 0$, and the matrix $A$ from (\ref{At}) takes the form
\begin{equation}
\label{A}
    A
    :=
    \bJ R_{\diam}  - \phi I_2/2
    =
    {\small
    \left[\begin{array}{cc}
         -\phi/2&  \ 1\\
        -\omega_0^2 - 12f(\kappa^2 + \sigma_{11}) & \  -\phi/2
    \end{array}\right]
    }.
\end{equation}
If $f> 0$, then the matrix $N$ in (\ref{best}) satisfies $N\succ 0$, so that the spectrum of $\bJ N$ is purely imaginary. If, in addition, $\phi>0$, then
$\det(\bJ N - \phi I_2/2)\ne 0$ and the first of the equations (\ref{infMS}) implies that $\mu = 0$ (in particular, $\kappa = 0$) for  the steady-state mean values of the system observables. Moreover, in this case, the matrix $A$ in (\ref{A}) is Hurwitz, thus leading to an admissible solution $\Sigma:= (\sigma_{jk})_{1\< j,k\< 2}$ (with $\Sigma + i\bJ/2 \succ 0$) of the second equation in (\ref{infMS}) to be found from
\begin{equation}
\label{Ric}
    \wh{A} \Sigma + \Sigma \wh{A}^{\rT} + C - 12f\sigma_{11}
        {\small
    \left[\begin{array}{cc}
        0 & \sigma_{11} \\
        \sigma_{11} & 2\sigma_{12}
    \end{array}\right]
    }
    =0,
\end{equation}
where
\begin{equation}
\label{Ahat}
    \wh{A}
    :=
    {\small
    \left[\begin{array}{cc}
        -\phi/2 & 1 \\
        -\omega_0^2 & -\phi/2
    \end{array}\right]
    }
\end{equation}
is a constant Hurwitz matrix. Due to the term, which is quadratic in $\Sigma$,  the structure of (\ref{Ric}) resembles that of an algebraic Riccati  equation. In view of (\ref{Ahat}), the matrix algebraic equation (\ref{Ric}), whose left-hand side is a real symmetric $(2\x 2)$-matrix,  is equivalent to three equations
\begin{eqnarray}
\label{Ric1}
    2\sigma_{12} -\phi\sigma_{11} + c_{11} & = & 0,\\
\label{Ric2}
    -\omega_0^2 \sigma_{11} -\phi \sigma_{12} + \sigma_{22} +  c_{12} - 12 f \sigma_{11}^2 & = & 0,\\
\label{Ric3}
    -2\omega_0^2 \sigma_{12} -\phi\sigma_{22} + c_{22} - 24 f\sigma_{11} \sigma_{12} & = & 0
\end{eqnarray}
with three unknowns $\sigma_{11}$, $\sigma_{12}$, $\sigma_{22}$.
By expressing $\sigma_{12}$ and $\sigma_{22}$ from (\ref{Ric1}) and (\ref{Ric3}) in terms of $\sigma_{11}$ as
\begin{eqnarray}
\label{sigma12}
    \sigma_{12} & = & (\phi \sigma_{11} - c_{11})/2,\\
\nonumber
    \sigma_{22}
    & = &
    \big(c_{22}-2(\omega_0^2 + 12 f\sigma_{11}) \sigma_{12}\big)/\phi\\
\label{sigma22}
    & = &
    \big((\omega_0^2 + 12 f\sigma_{11}) (c_{11}-\phi \sigma_{11}) + c_{22}\big)/\phi,
\end{eqnarray}
and substituting the representations into (\ref{Ric2}), it follows that $\sigma_{11}$ satisfies a quadratic equation:
\begin{equation}
\label{poly}
    24 f \sigma_{11}^2
    +
    (2\omega_0^2 + \phi^2/2-12 f c_{11}/\phi ) \sigma_{11}
    -
    ((\omega_0^2+\phi^2/2) c_{11} + c_{22} + \phi c_{12})/\phi
    =0.
\end{equation}
Since the matrix (\ref{C}) satisfies $C\succ 0$, so that $c_{11}c_{22}> c_{12}^2$, then, by the arithmetic-geometric mean inequality, $ \phi^2 c_{11}/2 + c_{22} \> \phi \sqrt{2c_{11}c_{22}}> \sqrt{2}\phi |c_{12}| $, and hence, the free term of the quadratic polynomial in (\ref{poly}) satisfies $(\omega_0^2+\phi^2/2) c_{11} + c_{22} + \phi c_{12} > (\sqrt{2}-1)\phi |c_{12}|\>0$. Therefore, in view of the assumption that $f>0$, the polynomial has two real roots with opposite signs, of which the positive root $\sigma_{11}$ makes the matrix $A$ in (\ref{A}) Hurwitz and yields a unique admissible solution $\Sigma$ for (\ref{Ric}) whose other entries are computed through (\ref{sigma12}) and (\ref{sigma22}). As a numerical example, suppose the open quantum Duffing oscillator is driven by a four-dimensional quantum Wiener process, and
\begin{equation}
\label{data}
    \omega_0 := 0.9026,
    \qquad
    B := {\small
    \left[\begin{array}{rrrr}
    0.4853 &  -0.1497 &   -0.0793 &   -0.6065\\
   -0.5955 &  -0.4348 &   1.5352  & -1.3474
    \end{array}\right]
    }.
\end{equation}
Here, (\ref{phi}) is satisfied with $\phi = 0.6357$. The behavior of the entries of the matrix $\Sigma$, computed by using the self-consistent quantum Gaussian linearization procedure for a range of nonnegative values of the anharmonicity parameter $f$, are shown in Fig.~\ref{fig:sigma_11_12_22}.
\begin{figure}[htbp]
\vskip-5cm\centerline{\includegraphics[width=14cm]{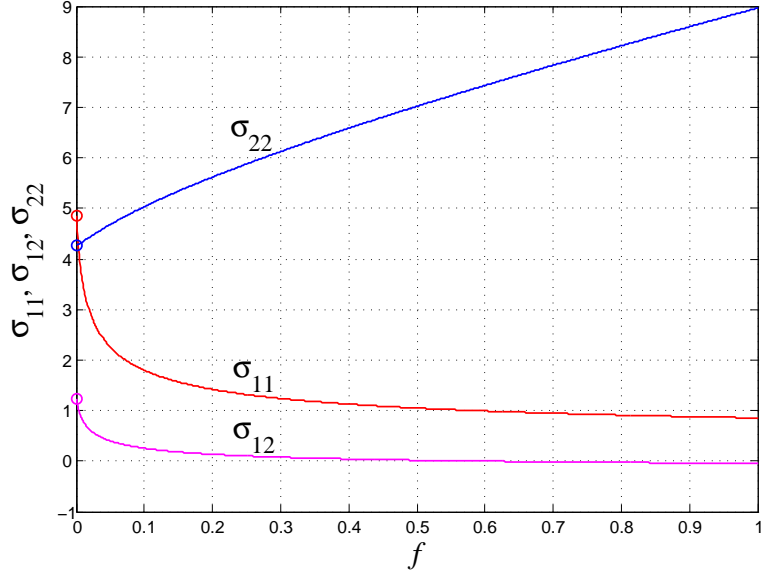}}
\vskip-5cm\caption{The entries $\sigma_{jk}$ of $\Sigma$, computed as an approximation to the real part of the steady-state quantum covariance matrix of the open quantum Duffing oscillator with (\ref{data}) via the self-consistent Gaussian linearization  for the range $0\< f\< 1$  of the anharmonicity strength  parameter in (\ref{HD}). The exact values of $\sigma_{jk}$ in the purely harmonic case $f=0$  are marked by ``$\circ$''s. }
\label{fig:sigma_11_12_22}
\end{figure}
The results of this approximation predict the decrease in the variance $\sigma_{11}$ of the position operator and the increase in the variance $\sigma_{22}$ of the  momentum operator as $f$ increases. This is in qualitative agreement with the fact that the quartic term $fq^4$ in the Hamiltonian (\ref{HD}) with large positive $f$ significantly ``steepens'' the walls of the potential well. In the case $f<0$,  when the minimum of the potential $V(q)$ at $q=0$ is only local,  the properties of the linearization are  more complicated and will be discussed elsewhere.

%
%
%

\section{Conclusion}\label{sec:conclusion}

We have proposed a quantum   Gaussian linearization  technique for a class of nonlinear open quantum systems with canonically commuting state variables governed by QSDEs with non-quadratic system Hamiltonians and bilinear coupling with the external heat bath. The approach is based on approximating the actual Hamiltonian by a quadratic function of the system observables in a mean square optimal fashion over a Gaussian reference quantum state. The optimal quadratic approximation of the Hamiltonian involves the inverse of a grade two special self-adjoint operator on real symmetric matrices, and we have described a method for its computation, more economical than that via vectorization of matrices.

The resulting differential equations for the approximations of the mean and quantum covariance matrix of the system observables  form a self-consistent set of equations which, despite their nonlinearity, produce a legitimate quantum covariance matrix. Moreover, they preserve the CCRs of the system observables, thus making the effective linear dynamics physically realizable.  We have demonstrated the approach for the quantum Duffing oscillator whose Hamiltonian is a quadro-quartic polynomial of the momentum and position operators.

A time invariant version of the proposed technique involves nonlinear vector-matrix algebraic equations for which the existence/uniqeness of admissible solutions is, in general, a nontrivial problem to be tackled elsewhere. The error analysis and detailed discussion of other properties of the quantum Gaussian linearization, and its applications to suboptimal filtering and control in nonlinear quantum systems, are also  intended for further publications.

%
%


\appendix
\section*{Appendices}

%
%

\section{Symmetries in the product moment of three observables}\label{sec:three}
\renewcommand{\theequation}{A\arabic{equation}}
\setcounter{equation}{0}

The following subsidiary lemma is used for establishing the symmetry of some of the mixed moments in (\ref{Hx_Hxx_xx}), (\ref{xxx_xxxx}) and is provided here for completeness of exposition.

\begin{lemma}
\label{lem:three}
Suppose $A$, $B$, $C$ are quantum observables on the underlying Hilbert space. Then the real part of their product  moment satisfies
\begin{equation}
\label{31}
    \Re
    \bE(ABC) =
    \Re
    \bE(CBA).
\end{equation}
Furthermore, if $A$ and $B$ (or $B$ and $C$) commute canonically, then
\begin{equation}
\label{32}
    \Re
    \bE(ABC) =
    \Re
    \bE(BAC)
\end{equation}
(respectively, $
    \Re
    \bE(ABC) =
    \Re
    \bE(ACB)
$). Finally, if each of the pairs $(A, B)$ and $(B,C)$ satisfies a CCR, then     $\Re
    \bE(ABC)$ is invariant under arbitrary permutations of $A$, $B$, $C$.
\end{lemma}
\begin{proof}
The first assertion (\ref{31}) of the lemma follows from the identities
\begin{equation}
\label{33}
    \overline{\bE(ABC)}
    =
    \overline{\Tr(\rho ABC)}
    =
    \Tr (CBA \rho)
    =
    \Tr(\rho CBA)
    =
    \bE(CBA),
\end{equation}
where $\rho$ is the density operator, and the self-adjointness of $A$, $B$, $C$ is used. To prove the second statement of the lemma, we note that if the observables $A$ and $B$ commute canonically, that is, if $[A,B]$ is a purely imaginary complex constant, then
$$
    \bE(ABC)-\bE(BAC) = \bE([A,B]C) = [A,B]\bE C
$$
is also purely imaginary by the realness of $\bE C$, thus implying   (\ref{32}).
 The other case (when $B$ and $C$ commute canonically)  is treated in a similar fashion, or alternatively, by reducing it to the  previous case through repeatedly using the invariance of $\Re \bE(ABC)$ under swapping $A$ with $C$. The third assertion of the lemma is established by combining the first two with the fact that the three transpositions $A\leftrightarrows C$, $A\leftrightarrows B$, $B\leftrightarrows C$  generate the group of all 6 possible permutations of $A$, $B$, $C$.
 \end{proof}
\section{Product moment of observables in a Gaussian quantum state}\label{sec:Wick}
\renewcommand{\theequation}{B\arabic{equation}}
\setcounter{equation}{0}

Let $\zeta:= (\zeta_j)_{1\< j \< n}$ be a vector of canonically commuting observables with a CCR matrix $\Theta \in \mA_n$, so that $[\zeta, \zeta^{\rT}] = i\Theta$. We will compute the product moment $\bE(\zeta_1\x\ldots \x \zeta_n)$ over a Gaussian quantum state in a different fashion from the traditional formulation and proof of Wick's theorem in terms of the annihilation and creation operators (\ref{faa}) and their normal ordering. To this end,
by repeatedly using the Baker-Hausdorff formula $\re^{A+B} = \re^A \re^{B-[A,B]/2}$ for operators $A$ and $B$ satisfying $[A,[A,B]] = [B,[A,B]] = 0$ (see, for example, \cite[pp. 128--129]{GZ_2004}), and the bilinearity of the commutator, it follows that
\begin{equation}
\label{multiBaker}
    \re^{i u^{\rT} \zeta}
     =
    \rprod_{k=1}^n
        \re^{iu_k \zeta_k-\left[\sum_{j=1}^{k-1} iu_j \zeta_j, iu_k \zeta_k\right]\big/2}
    =
    \re^{i\sum_{1\< j< k \< n}\theta_{jk} u_ju_k/2}
    \rprod_{k=1}^n
    \re^{iu_k \zeta_k}
\end{equation}
for any $u := (u_j)_{1\< j \< n} \in \mR^n$. Here, \
$
    \rprod_{k=1}^n
    A_k
    :=
    A_1\x \ldots  \x A_n
$ denotes
the product of operators $A_1, \ldots, A_n$, ordered  ``rightwards'', with the order of multiplication being important in the noncommutative case. If the system is in a Gaussian quantum state,  in which $\zeta$ has zero mean $\bE \zeta = 0$ and the quantum covariance matrix $S:= \bE(\zeta\zeta^{\rT})$, with  $\Im S = \Theta/2$, then the quantum characteristic function of $\zeta$ is
\begin{equation}
\label{cfzeta}
    \bE
    \re^{iu^{\rT}\zeta}
    =
    \re^{-u^{\rT} \Sigma u/2}
\end{equation}
for any $u \in \mR^n$,
where $\Sigma:= \Re S$; cf. (\ref{charfun}), (\ref{S}). By combining (\ref{cfzeta}) with (\ref{multiBaker}), it follows that
\begin{equation}
\label{Lambda}
    \Lambda(u)
     :=
    \bE
    \rprod_{k=1}^n
    \re^{iu_k \zeta_k}
    =
    \re^{-i\sum_{1\< j< k \< n}\theta_{jk} u_ju_k/2}
    \bE
    \re^{iu^{\rT} \zeta}
    =
    \re^{-u^{\rT} \wt{S} u/2}.
\end{equation}
Here,
\begin{equation}
\label{Stilde}
    \wt{S}:= (\wt{s}_{jk})_{1\< j,k\< n}:= \Sigma + i\wt{\Theta}/2
\end{equation}
is a complex symmetric matrix, where $\wt{\Theta}:= (\wt{\theta}_{jk})_{1\< j,k\< n}$ is a real symmetric matrix which is defined in terms of the CCR matrix $\Theta$ as
\begin{equation}
\label{thetatilde}
    \wt{\theta}_{kj}
    :=
    \wt{\theta}_{jk}
    :=
    \theta_{jk},
    \qquad
    1\< j\< k\< n,
\end{equation}
and use is made of the property that all the diagonal entries of $\Theta$ are zero. The product moment of the observables $\zeta_1, \ldots, \zeta_n$  in the Gaussian quantum state is obtained from the function $\Lambda$ in (\ref{Lambda}) as
\begin{equation}
\label{zetaprod}
    \bE \rprod_{k=1}^n
    \zeta_k
    =
    (-i)^n
    \d_{u_1}\ldots \d_{u_n}\Lambda(u)\big|_{u=0}.
\end{equation}
Although the matrix $\wt{S}$ in (\ref{Stilde}) is complex, its symmetry allows further analysis to be  carried out similarly to the classical Gaussian case, as if $\wt{S}$ were the covariance matrix of a vector of jointly Gaussian real-valued random variables. Indeed, since $\Lambda(-u) = \Lambda(u)$, all the odd order  partial derivatives of $\Lambda$ vanish at the origin, and  hence, the product moment $    \bE\ \rprod_{k=1}^n \zeta_k$ in (\ref{zetaprod}) can only be nonzero if $n$ is even, that is, if $n = 2r$ for some positive integer $r$. In this case, a degree $2r$ homogeneous polynomial $(-u^{\rT} \wt{S}u/2)^r/r!$ in $u_1, \ldots, u_{2r}$ is the only term of the Taylor series expansion $\Lambda(u) = \sum_{\ell\> 0} (-u^{\rT} \wt{S}u/2)^{\ell}/\ell!$ of the right-hand side of  (\ref{Lambda}) which contributes to the partial derivative in (\ref{zetaprod}). More precisely,
\begin{equation}
\label{sum}
    \d_{u_1}\ldots \d_{u_{2r}}\Lambda(u)\big|_{u=0}
     =
    \coeff_{u_1 \ldots u_{2r}}
    (-u^{\rT} \wt{S}u/2)^r/r!
    =
    \frac{(-1)^r}{r!2^r}
    \sum
    \prod_{m=1}^r
    \wt{s}_{j_mk_m},
\end{equation}
where $\coeff_{u_1 \ldots u_{2r}}(\cdot)$ denotes the coefficient of the polynomial associated with the monomial $u_1\x\ldots \x u_{2r}$, and the sum is taken over a class $\Ups_r$ of all possible $(2r)!$ permutations $(j_1, k_1, \ldots, j_r, k_r)$ of the integers $1, \ldots, 2r$. For any $(j_1, k_1, \ldots, j_r, k_r) \in \Ups_r$,  the product on the right-hand side of  (\ref{sum}) is invariant under $r!$ permutations of the pairs $(j_1, k_1)$, \ldots, $(j_r, k_r)$ and with respect to $r$  transpositions $j_1 \leftrightarrows k_1$, \ldots, $j_r \leftrightarrows k_r$, with the latter invariance following from the symmetry of the complex matrix $\wt{S}$ defined by (\ref{Stilde})--(\ref{thetatilde}). The action of a group, generated by the $r!$ pair permutations  and $r$ transpositions within each of the pairs, splits $\Ups_r$ into $(2r)!/(r!2^r) = (2r-1)!!$ equivalence classes. Each of these classes is the orbit of the group passing through one of the regular permutations $(j_1, k_1, \ldots, j_r,k_r)$ of the integers $1,\ldots, 2r$  which satisfies $j_1 < k_1$, \ldots, $j_r< k_r$ and $j_1 < j_2 <\ldots < j_{r-1} < j_r$ (see Section~\ref{sec:Gauss}). Therefore, the sum in (\ref{sum}) reduces to that over the class $\cP_r$ of regular permutations (as representatives of the equivalence classes) as
\begin{eqnarray}
\label{sum1}
    \d_{u_1}\ldots \d_{u_{2r}}\Lambda(u)\big|_{u=0}
      =
      (-1)^r
    \sum_{\cP_r}
    \prod_{m=1}^r
    \wt{s}_{j_mk_m}.
\end{eqnarray}
Now, since $\wt{s}_{jk} = s_{jk}$ for all $j\< k$ in view of (\ref{Stilde})--(\ref{thetatilde}), then $\prod_{m=1}^r
    \wt{s}_{j_mk_m} = \prod_{m=1}^r
    s_{j_mk_m}$ for any $(j_1, k_1, \ldots, j_r,k_r) \in \cP_r$. Hence, by combining (\ref{zetaprod}) (with $n=2r$) and (\ref{sum1}), it follows that
\begin{equation}
\label{zetaprod1}
    \bE \rprod_{k=1}^{2r}
    \zeta_k
    =
    \sum_{(j_1, k_1, \ldots, j_r,k_r) \in \cP_r}\,
    \prod_{m=1}^r
    s_{j_mk_m},
\end{equation}
which is what constitutes Wick's theorem. In the case $\Theta = 0$ of pairwise commuting observables, when all the cross-covariances are symmetric, that is, $s_{jk} = s_{kj}$,    the relation (\ref{zetaprod1}) reproduces Isserlis' theorem \cite{I_1918} for jointly Gaussian random variables, which, due to the symmetry,  is usually formulated in terms of pair partitions $\{\{j_1, k_1\}, \ldots, \{j_r, k_r\}\}$ of the integers $1, \ldots, 2r$. However, in the noncommutative quantum case, when $\Theta \ne 0$ and $s_{kj} = \overline{s_{jk}}$, the conditions $j_1< k_1$, \ldots, $j_r<k_r$ in the definition of regular permutations become essential for  (\ref{zetaprod1}).

\section{Proof of Lemma~\ref{lem:mom}}\label{sec:momproof}
\renewcommand{\theequation}{C\arabic{equation}}
\setcounter{equation}{0}

In the Gaussian quantum state,  the expectation on the left-hand side of (\ref{PhiPsi}) can be  computed  for any matrix $R \in \mS_n$  as
\begin{eqnarray}
\nonumber
    \bE((\xi^{\rT}R\xi)^2)
     & = &
    \sum_{j,k,\ell,m=1}^{n}
    r_{jk}r_{\ell m}
    \bE(\xi_j\xi_k\xi_{\ell} \xi_m)\\
\nonumber
     & = &
    \sum_{j,k,\ell,m=1}^{n}
        r_{jk}r_{\ell m}
    (s_{jk} s_{\ell m}
    +
    s_{j\ell} s_{k m}
    +
    s_{jm} s_{k\ell})\\
\nonumber
    & = &
    (\Tr(RS))^2
    +
    2\Tr (RS^{\rT}RS)=
    \bra
        R,
        \Sigma
    \ket^2
    +
    2\Tr (R(\Sigma R\Sigma + \Theta R \Theta/4))    \\
\label{PhiPsi1}
    & = &
    \bra
        R,
        \Sigma \bra \Sigma, R\ket
        +
        2K(R)
    \ket,
\end{eqnarray}
where we have used the symmetry of $R$ and $\Sigma$ and the antisymmetry of $\Theta$ (by which $\Tr (R\Sigma R \Theta) =0$), and also the representation (\ref{K})  of the operator $K(R):= \Re(S^{\rT} RS)$.
Comparison of the right-hand sides of (\ref{PhiPsi}) and (\ref{PhiPsi1})
leads to  (\ref{Psigauss}). We will now prove that $S\succ 0$ entails $K\succ 0$. Since the property $S\succ 0$ for the complex matrix (\ref{S}) is equivalent to the condition
$
    {\scriptsize\left[\begin{array}{cc}
        \Sigma & -\Theta/2\\
        \Theta/2 & \Sigma
    \end{array}
    \right]}
    \succ
    0
$
for the real matrices $\Sigma$ and $\Theta$, it is also equivalent to that $\Sigma \succ 0$ and that the real antisymmetric matrix \begin{equation}
\label{Xi}
    \Xi
    :=
    \Sigma^{-1/2} \Theta \Sigma^{-1/2}/2
\end{equation}
is contractive in the sense that its operator norm $\| \Xi\|_{\infty} := \sqrt{\lambda_{\max}(\Xi^{\rT} \Xi)}$ (to be distinguished from the Frobenius norm $\|\cdot\|$ of matrices, with $\lambda_{\max}(\cdot)$ the largest eigenvalue) satisfies
\begin{equation}
\label{Xismall}
    \|\Xi\|_{\infty}
    <
    1.
\end{equation}
Since $\Xi$ is antisymmetric, its operator norm coincides with the spectral radius: $\|\Xi\|_{\infty} = \br(\Xi)$.
By bijectively transforming the matrix $R$ into another real symmetric matrix
\begin{equation}
\label{Rnew}
    \wt{R}
    :=
    \sqrt{\Sigma} R \sqrt{\Sigma}
\end{equation}
(recall that $\Sigma \succ 0$, which ensures the existence of a real positive definite symmetric matrix square root $\sqrt{\Sigma}$)
and combining (\ref{K}) with (\ref{Xi}), it follows that
\begin{equation}
\label{RKR}
    \bra
        R, K(R)
    \ket
    =
    \bra
        \wt{R},
        \wt{R}
        +
        \Xi \wt{R}\Xi
    \ket=
    \|\wt{R}\|^2
    +
    \bra
        \wt{R},
        \Xi \wt{R} \Xi
    \ket
    \>
    (1-\|\Xi\|_{\infty}^2)
    \|\wt{R}\|^2.
\end{equation}
Here, we have also used the Cauchy-Bunyakovsky-Schwarz inequality (applied to the Frobenius inner product of matrices) and the property that neither $\|\Xi\wt{R}\|$ nor $\|\wt{R}\Xi\|$ exceeds $\|\Xi\|_{\infty} \|\wt{R}\|$, by which
$$
    \bra
        \wt{R},
        \Xi \wt{R} \Xi
    \ket
    =
    -
    \bra
        \Xi
        \wt{R},
         \wt{R} \Xi
    \ket
     \>
    -\|\Xi\wt{R}\| \|\wt{R}\Xi\|
    \>
    -\|\Xi\|_{\infty}^2 \|\wt{R}\|^2.
$$
Indeed,
$
    \|\Xi \wt{R}\|^2
    =
    \Tr (\wt{R} \Xi^{\rT} \Xi \wt{R})
     \<
    \lambda_{\max}(\Xi^{\rT} \Xi)
    \Tr(\wt{R}^2)
    =
    \|\Xi\|_{\infty}^2
    \|\wt{R}\|^2
$,
and the inequality $\|\wt{R} \Xi\| \< \|\Xi\|_{\infty}
    \|\wt{R}\|$ is verified in a similar fashion.
In view of (\ref{Xismall}) and arbitrariness of the matrix $R\in \mS_n$ in (\ref{Rnew}), the lower bound (\ref{RKR}) shows that the condition $S\succ 0$ does entail $K\succ 0$.

\section{Computing the inverse of the operator $K$ in  (\ref{K})}\label{sec:Kinv}
\renewcommand{\theequation}{D\arabic{equation}}
\setcounter{equation}{0}

Throughout this section, a short-hand notation
$
    \[[[
        \gamma_1, \delta_1 \mid \ldots \mid \gamma_r, \delta_r
    \]]]
    =
    \sum_{k=1}^r
    \[[[
        \gamma_k, \delta_k
    \]]]
$ will be utilized
for a \emph{special linear operator of grade} $r$ which acts on a matrix $X$ as
$$
    \[[[
        \gamma_1, \delta_1
        \mid
        \ldots
        \mid
        \gamma_r, \delta_r
    \]]](X)
    :=
    \sum_{k=1}^r
    \gamma_k X \delta_k,
$$
where $\gamma_1, \delta_1, \ldots, \gamma_r, \delta_r$ are given appropriately dimensioned real matrices.\footnote{Such operator structure resembles the Kraus form of quantum operations  \cite[pp. 360--373]{NC_2000}.} If for every $k=1,\ldots, r$, the matrices $\gamma_k$, $\delta_k$ are either both symmetric or both antisymmetric, $\[[[\gamma_1, \delta_1 \mid \ldots \mid \gamma_r, \delta_r\]]]$ is a self-adjoint operator whose properties are studied in \cite[Section 7]{VP_2011a}. Moreover, if for every $k$, the matrix $\gamma_k$ has the same order and is either symmetric or antisymmetric, then $\[[[\gamma_1,\gamma_1 \mid \ldots \mid \gamma_r,\gamma_r\]]]$ is a self-adjoint operator on an appropriate subspace of real symmetric matrices. By using the identity $\[[[\gamma,\delta\]]] \[[[\sigma,\tau\]]] = \[[[\gamma\sigma, \tau\delta \]]]$ for the  composition (interchangeably, product) of special operators of grade one, and recalling (\ref{Xi}), the operator $K$ from (\ref{K}) can be factorized as
\begin{equation}
\label{Kdecomp}
    K
    =
        \[[[
        \sqrt{\Sigma},
        \sqrt{\Sigma}
    \]]]
    \[[[
        I,I
        \mid
        \Xi, \Xi
    \]]]
    \[[[
        \sqrt{\Sigma},
        \sqrt{\Sigma}
    \]]].
\end{equation}
Here, $\[[[
        \sqrt{\Sigma},
        \sqrt{\Sigma}
    \]]] = \sqrt{\[[[
        \Sigma,
        \Sigma
    \]]]}
$
is a positive definite special  self-adjoint operator of grade one on the space $\mS_n$ which was employed in (\ref{Rnew}) to carry out the transformation $R\mapsto \wt{R}$. The latter operator is straightforwardly invertible, with
$
    \[[[
        \sqrt{\Sigma},
        \sqrt{\Sigma}
    \]]]^{-1}
    =
    \[[[
        \Sigma^{-1/2},
        \Sigma^{-1/2}
    \]]]
$
in view of $\Sigma \succ 0$. Therefore, (\ref{Kdecomp}) reduces
the computation of the inverse operator
\begin{equation}
\label{Kinv}
    K^{-1}
     =
    \[[[
        \Sigma^{-1/2},
        \Sigma^{-1/2}
    \]]]
    \[[[
        I,I
        \mid
        \Xi, \Xi
    \]]]^{-1}
    \[[[
        \Sigma^{-1/2},
        \Sigma^{-1/2}
    \]]]
\end{equation}
to that of $
\[[[
        I,I
        \mid
        \Xi, \Xi
    \]]]^{-1}
$.
Since the matrix $\Xi$, defined by  (\ref{Xi}), is real antisymmetric, it has purely imaginary spectrum $\pm i \omega_1, \ldots, \pm i\omega_{\nu}$ (with all $\omega_1, \ldots, \omega_{\nu}$ real and, without loss of generality, nonnegative) and is orthogonally block-diagonalizable in the sense that
\begin{equation}
\label{XiUU}
    \Xi
    =
    U
    (\bJ\ox \mho)
    U^{\rT}.
\end{equation}
Here, $U\in \mR^{n\x n}$ is an orthogonal matrix whose columns are formed from the real and imaginary parts of $\nu$ eigenvectors of $\Xi$ associated with the eigenvalues $i\omega_1, \ldots, i\omega_{\nu}$. Also, the matrix $\bJ$ is given by (\ref{bJ}), and
\begin{equation}
\label{mho}
    \mho
    :=
    \diag_{1\< k\< \nu}
    (\omega_k)
\end{equation}
is a diagonal matrix with $\omega_1, \ldots, \omega_{\nu}$ over the main diagonal, so that, with $\ox$ the Kronecker product of matrices,
\begin{equation}
\label{bJmho}
    \bJ
    \ox
    \mho
    =
    {\small\left[\begin{array}{cc}
        0 & \mho\\
        -\mho & 0
    \end{array}
    \right]}
\end{equation}
is a two-diagonal real antisymmetric matrix. Here, we prefer to utilize real matrices in order to avoid extensions of linear operators to spaces of complex matrices.
The orthogonality of the matrix $U$ implies that $\[[[ U, U^{\rT}\]]]$ is a unitary operator on the space $\mS_n$. Indeed, since its adjoint is $\[[[U,U^{\rT}\]]]^{\dagger} = \[[[U^{\rT}, U\]]]$, then
$$
    \[[[U,U^{\rT}\]]]^{\dagger}
    \[[[U,U^{\rT}\]]]
    =
    \[[[U^{\rT} U, U^{\rT}U\]]]
    =
    \[[[I,I\]]]
    =
    \cI
$$
is the identity operator on $\mS_n$. From (\ref{XiUU}) and the unitarity of $\[[[U,U^{\rT}\]]]$, it follows that
\begin{equation}
\label{IZ}
    \[[[
        U^{\rT},
        U
    \]]]
    \[[[
        I,I
        \mid
        \Xi, \Xi
    \]]]
    \[[[
        U,U^{\rT}
    \]]]
    =
    \cI +Z,
\end{equation}
where
\begin{equation}
\label{Z}
    Z
    :=
    \[[[
        \bJ\ox \mho,\,
        \bJ\ox \mho
    \]]]
\end{equation}
is a grade one self-adjoint operator on $\mS_n$, whose operator norm coincides with its spectral radius and is computed in terms of (\ref{mho})--(\ref{bJmho}) as
$$
    \br(Z)
    =
    \br(\bJ \ox \mho)^2
    =
    \max_{1\< k \< \nu} \omega_k^2.
$$
Therefore, since the matrix $\Xi$ is contractive in view of (\ref{Xismall}), so that $0\< \omega_1, \ldots, \omega_{\nu} <1$, then so also is the operator $Z$, thus ensuring the invertibility of the operator $\cI + Z$. Hence,  in view of the unitarity of $\[[[U^{\rT}, U\]]]$ in (\ref{IZ}), the inverse operator $K^{-1}$ in (\ref{Kinv}) can be computed as
\begin{eqnarray}
\nonumber
    K^{-1}
    & = &
    \[[[
        \Sigma^{-1/2},
        \Sigma^{-1/2}
    \]]]
    \[[[
        U,
        U^{\rT}
    \]]]
    (\cI + Z)^{-1}
    \[[[
        U^{\rT},U
    \]]]
    \[[[
        \Sigma^{-1/2},
        \Sigma^{-1/2}
    \]]]\\
\label{Kinv1}
    & = &
    \[[[
        \wt{U},
        \wt{U}^{\rT}
    \]]]
    (\cI + Z)^{-1}
    \[[[
        \wt{U}^{\rT},
        \wt{U}
    \]]].
\end{eqnarray}
Here,
\begin{equation}
\label{Utilde}
    \wt{U}
    :=
    \Sigma^{-1/2} U
\end{equation}
is a nonsingular matrix which satisfies $\Sigma^{-1}\Theta \wt{U}/2 = \wt{U}(\bJ \ox \mho)$ and is, therefore,  related to the eigenvectors of the matrix $\Sigma^{-1}\Theta/2 = \Sigma^{-1/2} \Xi \sqrt{\Sigma}$, isospectral to $\Xi$ from (\ref{Xi}). Now, to compute the inverse operator on the right-hand side of (\ref{Kinv1}), we factorize it as
\begin{equation}
\label{ZZZ}
    (\cI + Z)^{-1}
    =
    (\cI -Z^2)^{-1}
    (\cI - Z),
\end{equation}
where
\begin{equation}
\label{Z2}
    Z^2
    =
    \[[[
        \bJ^2 \ox \mho^2,
        \bJ^2 \ox \mho^2
    \]]]
    =
    \[[[
        D,D
    \]]],
\end{equation}
and
\begin{equation}
\label{D}
    D
    :=
    \diag_{1\< k \< n} (d_k)
    :=
    I_2 \ox \mho^2
    =
    {\small\left[\begin{array}{cc}
        \mho^2 & 0\\
        0 & \mho^2
    \end{array}
    \right]}
\end{equation}
is a diagonal matrix with diagonal entries $d_k = d_{k+\nu} = \omega_k^2$ for $k=1, \ldots, \nu$.
Here, we have used the property $\bJ^2=-I_2$ for the matrix (\ref{bJ}) and the power identities $\[[[\gamma,\delta\]]]^k = \[[[\gamma^k,\delta^k\]]]$ and $(\gamma \ox \delta)^k = \gamma^k \ox \delta^k$ for square matrices $\gamma$, $\delta$ and nonnegative integers $k$. It follows from (\ref{Z2})--(\ref{D}) that the inverse operator on the right-hand side of (\ref{ZZZ}) can be  expanded into an absolutely convergent operator power series as
\begin{equation}
\label{IZinv}
    (\cI - Z^2)^{-1}
    =
    \sum_{\ell\> 0}
    Z^{2\ell}
    =
    \sum_{\ell\>0}
    \[[[
        D^{\ell},D^{\ell}
    \]]].
\end{equation}
Due to the diagonal structure of the matrix $D$ from (\ref{D}), the image $(\cI - Z^2)^{-1}(Y)$ of a matrix $Y:= (y_{jk})_{1\< j,k\< n} \in \mS_n$ under the operator (\ref{IZinv}) is a real symmetric matrix with entries
$$
    ((\cI -Z^2)^{-1}(Y))_{jk}
    =
    \sum_{\ell \> 0}
    d_j^{\ell}
    y_{jk}d_k^{\ell}
    =
    y_{jk}/(1-d_jd_k)
$$
for all $1\< j,k\< n$.
Hence, upon splitting the matrix $Y:= {\scriptsize\left[\begin{array}{cc}Y_{11} & Y_{12}\\ Y_{21} & Y_{22}\end{array}\right]}$ into four $(\nu\x \nu )$-blocks $Y_{jk}$, its image takes the form
\begin{equation}
\label{I-Z2}
    (\cI -Z^2)^{-1}(Y)
     =
    {\small
    \left[\begin{array}{cc}
        L \od Y_{11} & L\od Y_{12}\\
        L\od Y_{21} & L\od Y_{22}
    \end{array}
    \right]}=
    (\bone_2 \ox L)\od  Y.
\end{equation}
Here, $\od$ denotes the Hadamard product of matrices, the matrix $L \in \mS_{\nu}$ is defined by
\begin{equation}
\label{L}
    L:= (1/(1-\omega_j^2\omega_k^2))_{1\< j,k\< \nu},
\end{equation}
and $\bone_r$ denotes the $(r\x r)$-matrix of ones. To represent the other factor  $\cI - Z$ on the right-hand side of (\ref{ZZZ}), we note that (\ref{mho}) and (\ref{bJmho}) allow the image of the matrix $Y$ under the operator $Z$ in (\ref{Z}) to be computed as
\begin{eqnarray}
\nonumber
    Z(Y)
     & = &
     {\small
     \left[\begin{array}{cc}
        0 & \mho\\
        -\mho & 0
     \end{array}
     \right]}
     Y
     {\small
     \left[\begin{array}{rc}
        0 & \mho\\
        -\mho & 0
     \end{array}
     \right]}=
    {\small
    \left[\begin{array}{cc}
        -\mho Y_{22} \mho &
         \mho Y_{21} \mho\\
         \mho Y_{12} \mho &
        -\mho Y_{11} \mho
    \end{array}
    \right]}\\
\label{ZY}
     & = &
    {\small
    \left[\begin{array}{cc}
        -M\od Y_{22} &  M\od Y_{21}\\
        M\od Y_{12} & -M\od Y_{11}
    \end{array}
    \right]}    =
    (\bone_2\ox M)
    \od
    ((\bJ\ox I_{\nu}) Y (\bJ\ox I_{\nu})),
\end{eqnarray}
where
\begin{equation}
\label{M}
    M:= (\omega_j\omega_k)_{1\< j,k\< \nu}
    =
    \omega\omega^{\rT},
    \qquad
    \omega
    :=
    {\small\left[\begin{array}{c}
        \omega_1\\
        \vdots\\
        \omega_{\nu}
    \end{array}
    \right]}.
\end{equation}
Finally, by assembling (\ref{Kinv1}), (\ref{ZZZ}), (\ref{I-Z2}), (\ref{ZY})  together, it follows that the inverse operator  $K^{-1}$ can be computed for the matrix $\Gamma\in \mS_n$ in (\ref{alphadiamgauss_betadiamgauss_Rdiamgauss})    as
\begin{equation}
\label{KinvGamma}
    K^{-1}(\Gamma)
     =
    \wt{U}
    (
        (\bone_2\ox L)
        \od
        (
            \wt{\Gamma}
            -
            (\bone_2\ox M)
            \od
            (
            (\bJ\ox I_{\nu}) \wt{\Gamma} (\bJ\ox I_{\nu})
            )
        )
    )\wt{U}^{\rT},
\end{equation}
with
        $\wt{\Gamma}
    :=
    \wt{U}^{\rT}\Gamma \wt{U}$.
Here, the matrices $\wt{U}$, $L$, $M$ are related by (\ref{Utilde}), (\ref{L}), (\ref{M}) with the eigenvectors and eigenvalues of the matrix $\Xi$ from (\ref{Xi}). It now remains to note that the positiveness (\ref{Kpos}) of $K^{-1}$ follows from (\ref{KinvGamma}) which represents this operator as the composition of positive operators $\[[[\wt{U}^{\rT}, \wt{U}\]]]$, $\[[[ \bJ^{\rT}\ox I_{\nu}, \bJ\ox I_{\nu}\]]]$, $\cI + (\bone_2\ox M)\odot $, $(\bone_2\ox L)\od $ and $\[[[\wt{U}, \wt{U}^{\rT}\]]]$. Here, we have used the positive semi-definiteness of the matrices $\bone_{\nu}$, $M$ and
$$
    L
    =
    \bone_{\nu}
    +
    \sum_{k\>1}
    \underbrace{M\od \ldots \od M}_{2k\ {\rm times}}
$$
and the fact that the Kronecker and Hadamard products of positive semi-definite matrices are also positive semi-definite; see the Schur product theorem \cite[Theorem 7.5.3 on p. 458]{HJ_2007}.

\section{Quadratic approximation of the Hamiltonian for the quantum Duffing oscillator}\label{sec:Duffcalcs}
\renewcommand{\theequation}{E\arabic{equation}}
\setcounter{equation}{0}

The mean square optimal quadratic approximation of the Hamiltonian $H$ in (\ref{HD}) reduces to the approximation of the quartic term $q^4$ by a quadratic function of $q$ and $p$. The latter problem reduces to the quadratic approximation of the non-quadratic part of
\begin{equation}
\label{q4}
    q^4
    =
    \kappa^4
    +
    4\kappa^3 \chi
    +
    6 \kappa^2 \chi^2
    +
        \lambda,
        \qquad
    \lambda
    :=
    \underbrace{4\kappa \chi^3 + \chi^4}_{\rm non-quadratic},
\end{equation}
 in the reference quantum state, with $\chi$ and $\varpi$ the centered position and momentum operators:
\begin{equation}
\label{chi_varpi}
    \xi
    =
    {\small\left[\begin{array}{c}
        \chi\\
        \varpi
    \end{array}
    \right]},
    \qquad
    \chi:= q-\kappa,
    \qquad
    \kappa:= \bE q,
    \qquad
    \varpi
    :=
    p-\bE p,
\end{equation}
where the centered vector $\xi$ of system observables is defined in conformance with (\ref{xieta}).
We will now compute the contribution of $\lambda$ from (\ref{q4}) to the quadratic approximation of the Hamiltonian (\ref{HD}) through the vector $\eps$ and the matrix $\Gamma$ from (\ref{Hx_Hxx_xx}) in a Gaussian reference quantum state:
\begin{equation}
\label{epsD_GammaD}
    \wh{\eps}
     =
    \Re \bE
    (\wh{\eta} \xi)
    =
    4\kappa
    \Re \bE
    (\chi^3\xi),
    \qquad
    \wh{\Gamma}
     =
    \Re \bE
    (\wh{\eta} \xi\xi^{\rT})
    =
    \Re \bE
    (\chi^4\xi\xi^{\rT})
    -3\sigma_{11}^2\Sigma.
\end{equation}
Here, in accordance with (\ref{xieta}) and (\ref{chi_varpi}), and by applying Wick's theorem (see Appendix~\ref{sec:Wick}) to $\bE(\chi^4) = 3(\bE(\chi^2))^2$,
\begin{equation}
\label{hateta}
    \wh{\eta}
    :=
    \lambda - \bE\lambda
    =
    \lambda-3\sigma_{11}^2
\end{equation}
is obtained by centering the non-quadratic part $\lambda$, with $\sigma_{11} := \bE(\chi^2)$ the variance of the position operator,
and use is made of the relations $\bE(\chi^4 \xi) = 0$ and $\bE(\chi^3 \xi\xi^{\rT}) = 0$  which follow from the property that the odd order central moments vanish in a Gaussian quantum state. Throughout this section, the ``hat'' symbol marks the quantities associated with the quadratic approximation of $\lambda$, such as  $\wh{\eps}$, $\wh{\Gamma}$, $\wh{\eta}$ in (\ref{epsD_GammaD}) and (\ref{hateta}). They will be multiplied by $f$ and combined with the remaining quadratic part of the Hamiltonian in (\ref{HD}). Application of Wick's theorem to the mixed moments in (\ref{epsD_GammaD}) yields
\begin{eqnarray}
\label{chi3xi}
    \bE(\chi^3 \xi)
    & = &
    3\bE(\chi^2) \bE(\chi \xi)
    =
    3\sigma_{11}
    {\small\left[\begin{array}{c}
        \bE(\chi^2)\\
        \bE(\chi\varpi)
    \end{array}
    \right]}
    =
    3\sigma_{11}
    {\small\left[\begin{array}{c}
        \sigma_{11}\\
        s_{12}
    \end{array}
    \right]},\\
\nonumber
    \bE(\chi^4 \xi_j\xi_k)
    & = &
    3
    (\bE(\chi^2))^2 \bE(\xi_j\xi_k)
    +
    12
    \bE(\chi^2) \bE(\chi \xi_j) \bE(\chi \xi_k)\\
\label{chi2xixi}
    & = &
    3\sigma_{11}^2 s_{jk} + 12 \sigma_{11} s_{1j} s_{1k}
\end{eqnarray}
for all $1\< j,k\< 2$, where $s_{jk}$ denote the entries of the quantum covariance matrix
\begin{equation}
\label{SDuff}
    S
    =
    {\small\left[\begin{array}{cc}
        s_{11} & s_{12}\\
        s_{21} & s_{22}
    \end{array}
    \right]}
    =
    {\small\left[\begin{array}{cc}
        \bE(\chi^2) & \bE(\chi\varpi)\\
        \bE(\varpi\chi) & \bE(\varpi^2)
    \end{array}
    \right]}
    =
    {\small\left[\begin{array}{cc}
        \sigma_{11} & \sigma_{12} + i/2\\
        \sigma_{12}-i/2 & \sigma_{22}
    \end{array}
    \right]}
\end{equation}
obtained from (\ref{bJ}), (\ref{S}), with $\Sigma:= (\sigma_{jk})_{1\< j,k\< 2} \in \mS_2^+$.
A vector-matrix form of (\ref{chi2xixi}) is
$
    \bE(\chi^4 \xi\xi^{\rT})
     =
    3
    \sigma_{11}^2 S
    +
    12
    \sigma_{11}
    {\scriptsize\left[\begin{array}{c}
        s_{11}\\
        s_{12}
    \end{array}
    \right]}
    {\small\left[
    \begin{array}{cc}
        s_{11} &
        s_{12}
    \end{array}
    \right]}
$.
The matrix $S$ from (\ref{SDuff}) satisfies the condition $S\succ 0$ of Theorem~\ref{th:bestquadgauss} if and only if $\sigma_{11}>0$,  $\sigma_{22}>0$ and
\begin{equation}
\label{detSigma}
    \det\Sigma
    =
    \sigma_{11}\sigma_{22}
    -
    \sigma_{12}^2
    >
    1/4,
\end{equation}
which is stronger than $\Sigma\succ 0$.
By combining (\ref{chi3xi})--(\ref{SDuff}), it follows that (\ref{epsD_GammaD}) take the form
\begin{equation}
\label{epsD1_GammaD1}
    \wh{\eps}
     =
    12\kappa
    \sigma_{11}
    \sigma,
    \qquad
    \wh{\Gamma}
    =
    3
    \sigma_{11}
    \left(
        4\sigma\sigma^{\rT}
        -
        {\small\left[\begin{array}{cc}
            0 & 0\\
            0 & 1
        \end{array}
        \right]}
    \right),
\end{equation}
where
$    \sigma
    :=
    {\scriptsize\left[\begin{array}{c}
        \sigma_{11}\\
        \sigma_{12}
    \end{array}
    \right]}
$
denotes the first column of the symmetric matrix $\Sigma$. Since $\Sigma^{-1}\sigma$ is the first column of $I_2$, substitution of (\ref{epsD1_GammaD1}) into (\ref{alphadiamgauss_betadiamgauss_Rdiamgauss}) yields
\begin{equation}
\label{betadiamDuff}
    \wh{\beta}_{\diam}
    =
    12\kappa \sigma_{11}
    \Sigma^{-1}
    \sigma
    =
    12 \kappa
    \sigma_{11}
    {\small\left[\begin{array}{c}
        1\\
        0
    \end{array}
    \right]},
\end{equation}
where $\kappa$ is the mean value of the position operator from (\ref{chi_varpi}).
We will now compute the matrix $R_{\diam}$ in (\ref{alphadiamgauss_betadiamgauss_Rdiamgauss}) by employing (\ref{KinvGamma}) of Appendix~\ref{sec:Kinv}. Since the matrix $\bJ$ from (\ref{bJ}) spans the space $\mA_2$,
the eigenvalues $\pm i \omega_1$ of the real antisymmetric $(2\x 2)$-matrix
\begin{equation}
\label{XiDuff}
    \Xi
    :=
    \Sigma^{-1/2}\bJ \Sigma^{-1/2}/2
    =
    \omega_1 U \bJ U^{\rT}
    =
    \omega_1
    \bJ
\end{equation}
from (\ref{Xi}) and (\ref{XiUU}) (with $U \in \mR^{2\x 2}$ an orthogonal matrix\footnote{The matrix $U$ is also symplectic with the structure matrix $\bJ$ in the sense that $U\bJ U^{\rT} = \bJ$.} being specified by the eigenvectors)  are given by
\begin{equation}
\label{om1}
    \omega_1
    =
    \frac{1}{2\sqrt{\det\Sigma}}
    <1,
\end{equation}
where the inequality follows (\ref{detSigma}). In the example being considered, the matrices $L$ and $M$ from (\ref{L}) and (\ref{M})  become positive scalars:
$
    L = 1/(1-\omega_1^4)
$
and
$
    M = \omega_1^2
$,
so that the action of each of the operators $(\bone_2 \ox L)\od $ and $(\bone_2 \ox M)\od $ on a $(2\x 2)$-matrix is equivalent to an appropriate scaling of the matrix. This allows (\ref{KinvGamma}), in application to (\ref{alphadiamgauss_betadiamgauss_Rdiamgauss}), to be simplified as
\begin{equation}
\label{RdiamDuff}
    \wh{R}_{\diam}
     =
     \frac{1}{1-\omega_1^4}
    \wt{U}
        (
            \wt{\Gamma}
            -
            \omega_1^2
            \bJ \wt{\Gamma} \bJ
    )\wt{U}^{\rT}
    =
     \frac{1}{1-\omega_1^4}
    \left(
        \Sigma^{-1} \wh{\Gamma} \Sigma^{-1}
        -
        \frac{\bJ \wh{\Gamma} \bJ}{(2\det\Sigma)^2}
    \right),
\end{equation}
with
\begin{equation}
\label{hatGamma}
    \Sigma^{-1} \wh{\Gamma} \Sigma^{-1}
    =
    3\sigma_{11}
    \left(
        4
        {\small\left[\begin{array}{cc}
            1 & 0\\
            0 & 0
        \end{array}
        \right]}
        -
        \Sigma^{-1}
        {\small\left[\begin{array}{cc}
            0 & 0\\
            0 & 1
        \end{array}
        \right]}
        \Sigma^{-1}
    \right)
\end{equation}
in view of (\ref{epsD1_GammaD1}).
Here, $        \wt{\Gamma}
    :=
    \wt{U}^{\rT}\wh{\Gamma} \wt{U}
$ and $\wt{U}:= \Sigma^{-1/2}U$  in accordance with (\ref{Utilde}), so that $\wt{U} \wt{U}^{\rT}=\Sigma^{-1}$ in view of the orthogonality of the matrix $U$ from (\ref{XiDuff}). Also, use has been made of the identity $Y^{-1} = -\bJ Y\bJ/\det Y$ which holds for any nonsingular symmetric matrix $Y$ of order two. Substitution of (\ref{epsD1_GammaD1}), (\ref{om1}), (\ref{hatGamma}) into (\ref{RdiamDuff}) yields
\begin{equation}
\label{RdiamDuff1}
    \wh{R}_{\diam}
    =
    12 \sigma_{11}
    {\small
    \left[
        \begin{array}{cc}
            1 & 0\\
            0 & 0
        \end{array}
    \right]
    }.
\end{equation}
In view of (\ref{betadiamDuff}) and (\ref{RdiamDuff1}), the mean square optimal quadratic approximation $\wh{\beta} \xi + \xi^{\rT}\wh{R}_{\diam} \xi/2$ of the operator $\lambda$ in (\ref{q4})  (with the additive constant terms being omitted) does not depend on the momentum operator. Now, by multiplying the quadratic approximation of $\lambda$ by $f$ and combining the result with the remaining quadratic part of the Hamiltonian (\ref{HD}), it follows that its mean square optimal quadratic approximation  in the Gaussian quantum state is
\begin{eqnarray*}
    H
    & \approx &
    \frac{p^2 + \omega_0^2q^2}{2}
    +
    f\big(
        4\kappa^3 \chi
        +
        6 \kappa^2 \chi^2
    +
    \wh{\beta}_{\diam}^{\rT}\xi
    +
    \xi^{\rT}\wh{R}_{\diam}\xi/2
    \big) + (*)\\
    & = &
    \beta_{\diam}^{\rT}\xi
    +
    \xi^{\rT}R_{\diam}\xi/2 + (*),
\end{eqnarray*}
where $(*)$ assembles the additive constant terms which are irrelevant for the dynamics of the quantum Duffing oscillator, and
\begin{eqnarray}
\label{betaDuff}
    \beta_{\diam}
    & = &
    {\small
    \left[
        \begin{array}{cc}
        \omega_0^2\kappa  + 4f \kappa^3\\
        \bE p
        \end{array}
    \right]}
    +
    f\wh{\beta}_{\diam}
    =
        {\small
    \left[
        \begin{array}{cc}
        (\omega_0^2  + 4f (\kappa^2 + 3\sigma_{11}))\kappa\\
        \bE p
        \end{array}
    \right]}, \\
\label{RDuff}
    R_{\diam}
    & = &
    {\small
    \left[
        \begin{array}{cc}
        \omega_0^2 + 12f\kappa^2 & 0\\
        0 & 1
        \end{array}
    \right]}
    +
    f\wh{R}_{\diam}
    =
    {\small
    \left[
        \begin{array}{cc}
        \omega_0^2 + 12f(\kappa^2+\sigma_{11}) & 0\\
        0 & 1
        \end{array}
    \right]}.
\end{eqnarray}

\end{document}